\documentclass[lettersize,journal]{IEEEtran}
\IEEEoverridecommandlockouts

\usepackage{url}
\usepackage{cite}
\usepackage{flushend}
\usepackage{makecell}
\usepackage{subfigure}
\usepackage{amsmath,amssymb,amsfonts}
\usepackage{amsthm}
\usepackage{algorithmic}
\usepackage{algorithm}

\usepackage{graphicx}
\usepackage{textcomp}
\usepackage{float}
\usepackage{booktabs}
\usepackage{array}
\usepackage{bm}
\usepackage{float} 
\usepackage{wasysym}
\usepackage{pifont}
\usepackage{comment}

\usepackage{cellspace}
\usepackage{multirow}
\usepackage[normalem]{ulem}
\useunder{\uline}{\ul}{}
\usepackage{xcolor}

\def\BibTeX{{\rm B\kern-.05em{\sc i\kern-.025em b}\kern-.08em
    T\kern-.1667em\lower.7ex\hbox{E}\kern-.125emX}}
\allowdisplaybreaks

\newtheorem{lemma}{Lemma}
\newtheorem{corollary}{Corollary}

\begin{document}

\title{
Communication-Efficient Soft Actor-Critic Policy Collaboration via Regulated Segment Mixture
\\ 
}
\author{Xiaoxue Yu, Rongpeng Li, Chengchao Liang, and Zhifeng Zhao\\

\thanks{X. Yu, and R. Li are with College of Information Science and Electronic Engineering, Zhejiang University (email: \{sdwhyxx, lirongpeng\}@zju.edu.cn).}
\thanks{C. Liang is with the School of Communication and Information Engineering, Chongqing University of Posts and Telecommunications (email: liangcc@cqupt.edu.cn).}
\thanks{Z. Zhao is with Zhejiang Lab as well as Zhejiang University (email: zhaozf@zhejianglab.com).}
\thanks{Part of the paper has been accepted by IEEE Globecom 2023 \cite{yu_communicationefficient_2023}.}
\thanks{Copyright (c) 20xx IEEE. Personal use of this material is permitted. However, permission to use this material for any other purposes must be obtained from the IEEE by sending a request to pubs-permissions@ieee.org.
}
}
\maketitle

\begin{abstract}
Multi-Agent Reinforcement Learning (MARL) has emerged as a foundational approach for addressing diverse, intelligent control tasks in various scenarios like the Internet of Vehicles, Internet of Things, and Unmanned Aerial Vehicles. 
However, the widely assumed existence of a central node for centralized, federated learning-assisted MARL might be impractical in highly dynamic environments. This can lead to excessive communication overhead, potentially overwhelming the system. 
To address these challenges, we design a novel communication-efficient, fully distributed algorithm for collaborative MARL under the frameworks of Soft Actor-Critic (SAC) and Decentralized Federated Learning (DFL), named RSM-MASAC. 
In particular, RSM-MASAC enhances multi-agent collaboration and prioritizes higher communication efficiency in dynamic systems by incorporating the concept of segmented aggregation in DFL and augmenting multiple model replicas from received neighboring policy segments, which are subsequently employed as reconstructed referential policies for mixing. 
Distinctively diverging from traditional RL approaches, RSM-MASAC introduces new bounds under the framework of Maximum Entropy Reinforcement Learning (MERL). Correspondingly, it adopts a theory-guided mixture metric to regulate the selection of contributive referential policies, thus guaranteeing soft policy improvement during the communication-assisted mixing phase. 
Finally, the extensive simulations in mixed-autonomy traffic control scenarios verify the effectiveness and superiority of our algorithm. 
\end{abstract}

\vspace{-1.1em}
\section{Introduction}\label{sec1}
Recently, Deep Reinforcement Learning (DRL) has gained significant traction in addressing complex, real-world applications, contingent on formulated Markov Decision Processes (MDPs) \cite{du2022comfortable, lin2020comparison}. 
Naturally, multiple collaborative DRL agents can form a scalable Multi-Agent Reinforcement Learning (MARL)-empowered system, efficiently coordinating sequential decision-making for complex scenarios \cite{li2022applications, nguyen2020deep, 9351818, shi2023deep, he2022multiagenta}. Notable examples include fleet management and traffic control in the Internet of Vehicles (IoV), Unmanned Aerial Vehicles (UAV) and multi-robot systems, as well as distributed resource allocation in the Internet of Things (IoTs).

\vspace{-1em}
\subsection{Problem Statement and Motivation}
In general, most MARL works adopt a Centralised Training and Decentralised Execution (CTDE) architecture \cite{MADDPG, MAPPO, COMA, DIAL_RIAL}.  
Federated Reinforcement Learning (FRL) \cite{FRL2021techniques, xu2021gradient, xie2023fedkl, 9945653} incorporates the continuous learning process in DRL with periodic model updates (i.e., exchanging gradients or parameters) from Federated Learning (FL), thus more competently coping with generalization difficulties during deployment \cite{han2023multi, leibo2021scalable}, such as variations in both physical and social environments. 
By keeping sensitive information localized, FRL enables instantaneous decision-making and promotes collaborative learning among independent agents. 
However, in most highly dynamic scenarios with high mobility, intermittent connectivity, and decentralized nature \cite{taik2022clustered, chellapandi2023federated}, the common assumption of a super-centralized training controller in Centralized Federated Learning (CFL) becomes impractical, threatening both the stability and timeliness of overall learning performance \cite{movahedian2023adaptive}. 
Besides, generalization difficulties \cite{meltingpotDeepmind, leibo2021scalable}, such as variations in external environments, necessitate continuous training or fine-tuning of decision-making models for specific environments or tasks with minimal online data requirements during deployment. 

In this way, there is growing research on peer-to-peer architecture's Decentralized Federated Learning (DFL) \cite{DFLsavazzi2020federated}, focusing on real-time decentralized training/fine-tuning of Deep Neural Network (DNN) models among DRL agents \cite{pacheco2024efficient, barbieri2022communication, nguyen2022deep}. 
This can be viewed as a distributed Stochastic Gradient Descent (SGD) optimization problem, in which the aggregation of exchanged DNN parameters acquired from the periodic communication phase typically uses a simplistic parameter average mixture approach. 
In FRL, these exchanged and averaged parameters pertain to the DRL agent's policy network, which approximates the policy distribution through a parameterized DNN. 
In addition, the parameter exchange can be implemented via Device-to-Device (D2D) or Vehicle-to-Vehicle (V2V) collaboration channels within communication range, common in many MARL works \cite{han2023multi, chen2024communication, shi2023deep, qu2024model}. 
However, the frequent information exchanges generate substantial communication overhead as the number of agents increases. 
Some DFL variants improve communication efficiency by reducing the number of communication rounds \cite{wang2021cooperative, liu2022decentralized, sun2022decentralized} or the communication workload per round \cite{Ako2016partial_gradient_exchange, barbieri2022communication, barbieri2023layer, hu2019decentralized}, but this can intensify the variability of local model updates, potentially leading to an inferior aggregated model post simple model parameter averaging \cite{kairouz2021advances}. 
This issue is more pronounced in online DRL frameworks since DRL agents interact more frequently with the environment compared to offline supervised learning. 
Incremental data can magnify learning discrepancies among agents and reduce fault tolerance. 
Interestingly, few studies have focused on the policy performance improvement issue of DRL with its combination with FL. 
Despite the simplicity and straightforward aggregation of mixed policies in parameter averaging, not all communicated packets in MARL contribute effectively. 
During the communication-assisted mixing phase, MARL awaits a revolutionary policy parameter mixture method and corresponding metrics to regulate the aggregation of exchanged model updates, ensuring robust policy improvement for security and safety. 

On the other hand, there is a contradiction between limited and costly sample availability and the need for fast learning speeds, that is, MARL algorithms shall fully explore and adapt to dynamic environments while minimizing the demand for extensive and costly online learning data.  
This aligns with the principles of Maximum Entropy Reinforcement Learning (MERL), such as Soft $Q$-learning \cite{softQlearning} and Soft Actor-Critic (SAC) \cite{sac1, sac2}, which transforms the traditional reward maximum into both expected return and the expected entropy of the policy. 
The incorporated entropy in optimization redefines value function and policy optimization objective from the ground up. Albeit its enhancement to sample efficiency, adaptability, and robustness in dynamic and uncertain environments, it overturns the proof of the traditional performance improvement bound \cite{xu2022trustable, kakade2002approximately, kuba2022trust, schulman2015trust}. 
Among MERL algorithms, SAC \cite{sac1, sac2} is particularly well-suited for highly dynamic and uncertain environments due to multi-folded reasons. First, it leads to 
higher sample efficiency through off-policy learning. Second, it enables automatic entropy adjustment via optimization of the temperature parameter. 
Finally, by effectively combining value function learning with policy optimization, it significantly enhances stabilization and adaptability to continuous, complex, high-dimensional action space.

Therefore, this work is dedicated to reanalyzing efficient communication within the MERL and DFL framework. Prominently, a practical policy mixture method, which is underpinned by a rigorously derived mixture metric, is devised to enable SAC to achieve smooth and reliable parameter updates, not only during independent local learning phase but also throughout the communication-assisted mixing phase. 

\begin{figure}[t]
\centering 
\includegraphics[scale = 0.45]{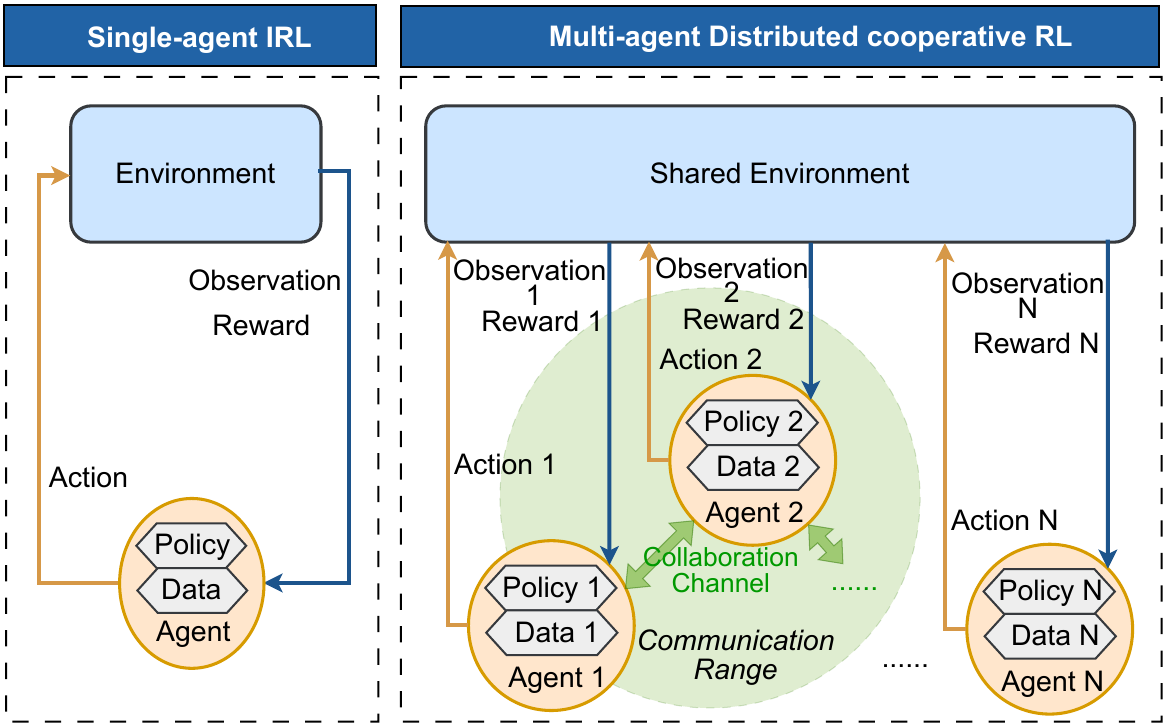}
\caption{Illustration of Single-agent IRL and Multi-agent Distributed Cooperative RL. }
\label{fig:system}
\vspace{-0.8em}
\end{figure}

\begin{table*}[tp]
    \centering
    \caption{{The Summary of Differences with Related Literature.}}
    \label{tab:related work}
    \def\arraystretch{1.05} 
    \hspace*{-0.2cm}
    \begin{tabular}{>{\raggedright\arraybackslash}p{1.8cm}|
    >{\centering\arraybackslash}m{1.2cm}
    >{\centering\arraybackslash}m{2.4cm}
    >{\centering\arraybackslash}m{2cm}
    >{\centering\arraybackslash}m{1.85cm}|m{6.2cm}}
\hline
           References & Maximum Entropy & Policy Improvement Guarantee &Collaboration via Communication   & Efficient Communication   & Brief Description      \\
\hline
{\cite{softQlearning, sac1,sac2}}   
& $\CIRCLE$   & $\Circle$ & $\Circle$ & $\Circle$  & \multirow{2}{*}{Only single-agent RL algorithm}\\
{\cite{schulman2015trust,kakade2002approximately}}
& $\Circle$  & $\CIRCLE$ & $\Circle$ & $\Circle$ &  \\
\cline{1-6}
{\cite{DVFs, change_value_functions,liu2017distributed}}
& $\Circle$  & $\Circle$ & $\CIRCLE$ & $\Circle$ & Over frequent and complex communication \\
\hline
{\cite{wang2021cooperative,
liu2022decentralized, 
sun2022decentralized,
Ako2016partial_gradient_exchange,
barbieri2022communication,
barbieri2023layer,
hu2019decentralized}}
& \multirow{2}{*}{$\Circle$}   & \multirow{2}{*}{$\Circle$} & \multirow{2}{*}{$\CIRCLE$ } & \multirow{2}{*}{$\CIRCLE$} & \multirow{2}{6.05cm}{Oversimplified parameter mixture method} \\
{\cite{xia2022convergence,
xing2021federated}} & & & & & \\
\hline
{\cite{kuba2022trust,xu2022trustable}}     
& $\Circle$   & $\CIRCLE$ & $\CIRCLE$ & $\Circle$  & \multirow{2}{6.3cm}{Only under traditional policy iteration-based approxi- \newline 
mately optimal RL's performance guarantee analysis}
 \\
{\cite{yu_communicationefficient_2023}}     
 & $\Circle$   & $\CIRCLE$ & $\CIRCLE$ & $\CIRCLE$  & \\
\hline
\textbf{This work} & $\CIRCLE$    & $\CIRCLE$ & $\CIRCLE$ & $\CIRCLE$ & Combining communication efficient DFL into MARL collaboration under maximum entropy framework with theory-established regulated mixture metrics and performance improvement bound \\
\hline
    \multicolumn{6}{>{\footnotesize\itshape}r}{Notations: \rm{${\Circle}$} \emph{indicates not included;} \rm{${\CIRCLE}$} \emph{indicates fully included.}}
    \end{tabular}
\end{table*}

\subsection{Contribution}
In this paper, we propose the Regulated Segment Mixture-based Multi-Agent SAC (RSM-MASAC) algorithm, tailored for training DRL agents under highly dynamic scenarios while addressing communication overhead challenges inherent in DFL. Our primary contributions include: 
 \begin{itemize}
    \item 
     For the highly dynamic setting, the proposed RSM-MASAC algorithm effectively combines communication-efficient DFL with MERL. The algorithm enables agents to receive segments of policy networks' parameters from neighbors within their communication range. These segments are used to constitute referential policies for strategically designed selective parameter mixture, ensuring credible performance improvement during the communication-assisted mixing phase while maintaining sufficient exploration in the local learning phase. 
    \item 
    In terms of mixing a current policy and a referential policy, 
    we derive a new, more generalized mixed policy improvement bound, which successfully tackles the analysis difficulties arising from the incorporation of redefined soft value function, dual policy optimization objective, and the logarithmic term of the policy for entropy maximum in MERL. Therefore, our work sets the stage for theoretically evaluating the performance of the mixed policy under distributed SGD optimization. 
    \item
    Instead of a simple parameter average in DFL, we bridge the relationship between DNN parameter gradient descent and MERL policy improvement, and regulate the selection of contributive referential policies by deriving a manageable, theory-guided mixture metric. Hence, it enhances the stability and practicality of directly mixing policy parameters during the communication-assisted mixing phase. 
    \item 
    Through extensive simulations in the traffic speed control task, a typical MARL-based IoV scenario, our proposed RSM-MASAC algorithm could approach the converged performance of centralized FMARL \cite{xu2021gradient} in a distributed manner, outperforming parameter average methods as in DFL \cite{hu2019decentralized,barbieri2022communication}, thus confirming its effectiveness. 
\end{itemize}

\subsection{Related Works}

In Multi-Agent Systems (MAS), Independent Reinforcement Learning (IRL) \cite{matignon2012independent, tan1993multi,papoudakis2020benchmarking}, which relies solely on agents' local perceptions without any collaboration, has been extensively studied and often used as a baseline due to its varied policy performance, unstable learning and uncertain convergence \cite{xu2021gradient}. 
Examples of IRL methods include IQL \cite{tan1993multi}, IAC \cite{COMA}, IA2C \cite{dhariwal2017openai} and IPPO \cite{IPPO}, etc. 

As an extension of IRL, distributed cooperative Reinforcement Learning (RL) enhances agents' collective capabilities and efficiency by collaboratively seeking near-optimal solutions through limited information exchange with others, as illustrated in Fig. \ref{fig:system}. 
Notably, the exchanged contexts can be rather different and possibly include approximated value functions in \cite{DVFs, change_value_functions}, rewards or even maximal $Q$-values on each state-action pair in \cite{liu2017distributed}. 
Besides, given the clear evidence \cite{sartoretti2019distributed, A2C2016asynchronous} that experiences from homogeneous, independent learning agents in MAS can contribute to efficiently learning a commonly shared DNN model, the direct exchange of model updates during FL communication phase \cite{FRL2021techniques, xu2021gradient, xu2021stigmergic} is a viable way to indirectly integrate information and enhance cooperation among IRL. 
In addition to the commonly used centralized architecture that is less suitable for actual dynamic environments, the peer-to-peer DFL paradigm, where clients exchange their local model updates only with their neighbors to achieve model consensus, emerges as an appealing alternative.
This approach can be viewed as distributed SGD optimization, where the significant communication expenditure cannot be overlooked. 
In that regard, some researchers have developed strategies to reduce communication frequency by aggregating more local updates before one round communication \cite{wang2021cooperative} or multi-round communication \cite{liu2022decentralized, sun2022decentralized}. 
Besides, reducing the number of parameters transmitted from local models or implementing selective model synchronization is also tractable. 
For instance, Ref. \cite{Ako2016partial_gradient_exchange} divides local gradients into several disjoint partitions, with only a subset being exchanged in any given communication round. 
Ref. \cite{barbieri2022communication} puts forward a randomized selection scheme for forwarding subsets of local model parameters to their one-hop neighbors. 
Furthermore, Ref. \cite{barbieri2023layer} suggests propagating top-k layers with higher normalized squared gradients, which may convey more information about the local data to neighbors.
Meanwhile, Ref. \cite{hu2019decentralized} introduces a segmented gossip approach that involves synchronizing only model segments, thereby substantially splitting the communication expenditure. 

Moreover, while the above distributed SGD optimization methods and most works in actual IoV, IoT and UAV applications \cite{taik2022clustered, chellapandi2023federated, pacheco2024efficient,camelo2019IoT, li2022applications} assume ideal communication between cooperative agents, achieving exact and perfect information sharing without compromising privacy concerns may not be feasible with suboptimal wireless transmission links. 
Beyond complex, customized, and costly cryptographic schemes \cite{kairouz2021advances, gao2023privacy}, most research \cite{reisizadeh2020fedpaq, xia2022convergence, xing2021federated} suggests local quantization of data before transmission, effectively shielding raw data from exposure.
Additionally, quantization, along with compression or sparsification \cite{compression2018communication, sun2022decentralized, tang2022gossipfl}, can significantly reduce the message size, with corresponding convergence analyses detailed in \cite{xia2022convergence, liu2022decentralized,xing2021federated}.

Notably, when it comes to the method of mixing DNN parameters, these aforementioned DFL works generally adopt a simplistic averaging approach, which is familiar in parallel distributed SGD methods. 
However, when such an approach is directly applied to FRL, the crucial relationship between parameter gradient descent and policy improvement will be overshadowed \cite{xu2022trustable}. 
Consequently, the corresponding mixture lacks proper, solid assessment means to prevent potential harm to policy performance \cite{xie2023fedkl}. 
In other words, directly using this kind of naive combination of communication efficient DFL and RL to enhance individual policy performance in IRL appears inefficient. 

Regarding this issue, we can draw inspiration from the conservative policy iteration algorithm \cite{kakade2002approximately}, which leverages the concept of policy advantage as a crucial indicator to gauge the cumulative reward improvement and applies a direct mixture update rule for policy distributions in pursuit of an approximately optimal policy.
Moreover, the mixture metric utilized in the update rule is also investigated to prevent overly aggressive updates towards risky directions, as excessively large policy updates often lead to significant policy performance deterioration \cite{xie2023fedkl}.
TRPO \cite{schulman2015trust} substitutes the mixture metric with Kullback-Leibler (KL) divergence, a measure that quantifies the disparity between current and updated policy distributions, facilitating the learning of monotonically improving policies. 
Ref. \cite{kuba2022trust} extends this work into cooperative MARL settings. 
However, directly mixing policy distributions is often an intractable endeavor. To implement this procedure in practical settings with parameterized policies in DRL, 
Ref. \cite{xu2022trustable} further simplifies the KL divergence to the parameter space through Fisher Information Matrix (FIM), so as to improve the policy performance by directly mixing DNN parameters. 
It focuses on stable policy updates throughout the communication-assisted mixing phase, but its analysis is confined to the traditional policy iteration-based RL algorithms, which aim to maximize the expected return only. 

Transitioning to the realm of MERL, the integration of the entropy maximization marks a significant paradigm shift, fundamentally redefining the criteria for policy improvement due to its dual objective of optimizing cumulative rewards and maintaining a high level of exploration. 
However, traditional RL methodologies, particularly those based on PPO, are insufficient to address the new dynamics interplayed between reward maximization and exploratory behavior imposed by the entropy maximization in MERL, as demonstrated in many works \cite{gao2023joint, duan2022distributional, sac1, sac2}. 
On the other hand, it's also imperative to implement an appropriate and manageable mixture metric with monotonic policy improvement property to maximize the practicality of directly mixing policy parameters during the communication-assisted mixing phase. 
Such a metric must guarantee the efficacy of the resultant mixed policy in terms of policy improvement. 
Crucially, the expected benefits of the mixed policy must be assessed before proceeding with actual policy mixing. Specifically, the mixed policy improvement should be evaluated against the referential policy, whose parameters are received from neighbors in DFL. 
Only after the anticipated benefits are confirmed should the operation to evaluate the established metric for mixing DNN parameters commence. 
Moreover, we have also summarized the key differences between our algorithm and relevant literature in Table \ref{tab:related work}. 
In conclusion, it is vital to reformulate the theoretical analysis within the combination of DFL and MERL framework. More particularly, given its superiority among MERL algorithms, SAC \cite{sac1, sac2} is chosen as a showcase offering robust theoretical underpinnings for practical IoV, IoT, and UAV applications.

\subsection{Paper Organization}

The remainder of this paper is organized as follows. 
In Section \ref{sec:Preliminary}, we present preliminaries of SAC algorithm and main notations used in this paper.
Then, we introduce the system model and formulate the problem in Section \ref{sec:system model and problem formulation}.
Afterward, we provide the mixed performance improvement bound theorem of FRL communication under MERL in Section \ref{sec:bound}, and elaborate on the details of the proposed RSM-MASAC algorithm in Section \ref{sec:MASAC design}.
In Section \ref{sec:simulations}, we present the simulation settings and discuss the experimental results.
Finally, Section \ref{sec:conclusion} concludes this paper and discusses future work.

\section{Preliminary}\label{sec:Preliminary}
Beforehand, we summarize the mainly used notations in Table \ref{n_table}.

\begin{table}[t]
	\centering
	\renewcommand{\arraystretch}{0.9} 
	\setlength{\extrarowheight}{1.5pt} 
	\caption{{Major notations used in the paper.}}
	\label{n_table}
	\begin{tabular}{|>{\bfseries}c>{\raggedright\arraybackslash}p{6.2cm}|}
		\hline
		\underline{\textbf{Notation}} & 
		\rule{0pt}{2.3ex} 
		\underline{\textbf{Definition}} \\ 
		\rule{0pt}{3ex} 
		$s_t^{(i)},a_t^{(i)},r^{(i)}_t$ & 
		Local state, individual action and reward of agent $i$ at time step $t$.\\
		$\pi, \tilde{\pi}, \pi_\text{mix}$
        & Current target policy distribution, referential target policy distribution and the mixed policy distribution, which represents the probabilities of selecting each possible action given a state.\\
	  $\theta, \tilde{\theta}, \theta_{\text{mix}}$ & The parameter of specific neural network that parameterizes the target policy distribution, referential target policy distribution and the mixed policy distribution, respectively. \\
		$p,P$ & Index of segments and segmentation granularity, $p=1,2,\cdots,P$.\\
		$\Omega_i$ & Set of one-hop neighbors within the communication range of agent $i$.\\
		$\zeta$ & Mixture metric of current policy parameter vector and referential policy parameter vector. \\
		$\alpha$ & Temperature parameter of policy entropy. \\
		$\varrho$ & Target smoothing coefficient of target $Q$ networks.\\
		$\kappa$ & Predefined model replica number. \\
		$U$ & Communication interval determined by specified iterations of the local policy.\\
		$\upsilon$ & Transmission bits of the policy parameters. \\
		$\psi$ & 	\rule[-1.5ex]{-2.6pt}{3ex}
		Communication consumption. \\\hline
	\end{tabular}
\end{table}

We consider the standard RL setting, where the learning process of each agent can be formulated as an MDP. During the interaction with the environment, at each time step $t$, an RL agent observes a local state $s_t$ from local state space $\mathcal{S}$, and chooses an action $a_t$ from individual action space $\mathcal{A}$ according to the policy $\pi(\cdot\vert s_t) \in \Pi$, which specifies a conditional distribution of all possible actions given the current state $s_t$, and $\Pi$ is the policy space. 
Then the agent receives an individual reward $r_t$, calculated by a reward function $\mathcal{R}:\mathcal{S}\times\mathcal{A}\rightarrow [0,R]$, and the environment transforms to a next state $s_{t+1}\sim p(\cdot \vert s_t,a_t)$. 
A trajectory starting from $s_t$ is denoted as $\tau_t = (s_t,a_t,s_{t+1},a_{t+1},\cdots)$. 
Besides, considering an infinite-horizon discounted MDP, the visitation probability of a certain state $s$ under the policy $\pi$ can be summarized as $d_\pi(s)=\sum\nolimits_{t=0}^\infty \gamma^t P(s_t=s;\pi)$, 
where $ P(s_t=s;\pi)$ is the visitation probability of the state $s$ at time $t$ under policy $\pi$.

Different from traditional RL algorithms aiming to maximize the discounted expected total rewards only, SAC \cite{sac1, sac2} additionally seeks to enhance the expected policy entropy. To be specific, with the re-defined soft state value function $V^\pi(s_0):= \mathop{\mathbb{E}}\limits_{\tau_0\sim\pi}\left[\sum_{t=0}^{\infty}\gamma^t(r_t+\alpha H(\pi(\cdot \vert s_t)))\right]$ on top of the policy entropy $H(\pi(\cdot \vert s))=-\mathbb{E}_{a\sim \pi}\log \pi(a\vert s)$, the task objective of SAC can be formulated as
\begin{equation}\label{eq:MERL objective}
\vspace{-0.2em}
    \max \eta(\pi) = 
    \mathop{\mathbb{E}}_{s_0\sim \rho_0}\left[V^\pi(s_0)\right] ,
\end{equation}
where $\rho_0$ is the distribution of initial state $s_0$ and $\alpha\in(0,\infty)$ is a temperature parameter determining the relative importance of the entropy term versus the reward. 
Obviously, when $\alpha\rightarrow 0$, SAC gradually approaches the traditional RL. Meanwhile, the soft state-action value can be expressed as $Q^{\pi}(s_t,a_t) := r_t+\mathop{\mathbb{E}}_{\tau_{t+1}\sim\pi}\!\left[\sum_{l=t+1}^{\infty}\gamma^{l-t} \Big(r_l \!+\!\alpha H(\pi(\cdot \vert s_{l}))\Big)\!\right]$ \cite{sac1, sac2}, which does not include the policy entropy of current time step but includes the sum of all future policy entropy and the sum of all current and future rewards. 
Consistently, the state-action advantage value under policy $\pi$ is 
\begin{equation}\label{eq:A definition}
    A_\pi(s_t,a_t)= Q^\pi(s_t,a_t)-V^\pi(s_t)  {.}
\end{equation}

Accordingly, SAC maximizes \eqref{eq:MERL objective} based on soft policy iteration, which alternates between soft policy evaluation and soft policy improvement. 
\begin{itemize}
    \item 
For given $\pi$, soft policy evaluation 
implies that $Q^\pi$ is learned by repeatedly applying soft Bellman operator $\mathcal{T}^\pi$ to the real-valued estimate $Q$, given by: 
\begin{align}
    \mathcal{T}^{\pi} Q(s_t,a_t) = r_t + \gamma \mathbb{E}_{s_{t+1}\sim p(\cdot\vert s_t,a_t)}[V(s_{t+1})],\label{eq:Q equation}
\end{align} 
where
\begin{align} V(s_t)&=\mathbb{E}_{a_t\sim\pi}[Q(s_t,a_t)]+\alpha  H(\pi(\cdot\vert s_t))  {.} \label{eq:V-function bellman equation}\end{align}
With $Q^{k+1}=\mathcal{T^\pi}Q^k$, as $k\rightarrow \infty$, $Q^k$ will converge to the soft $Q$ function $Q^\pi$ of $\pi$, as proven in \cite{sac2}. 

\item In the soft policy improvement step, the goal is to find a policy $\pi_\text{new}$ superior to the current policy $\pi_\text{old}$, in terms of maximizing \eqref{eq:MERL objective}. Specifically, for each state, SAC updates the policy as
\begin{align}\label{eq:SAC policy KL}
    \pi_{\mathrm{new}}=&\mathop{\arg\min}_{\pi\in\Pi} \mathrm{D}_{\mathrm{KL}}\Big( \pi(\cdot\vert s_t)\Vert \frac{ \exp\big( \frac{1}{\alpha} Q^{\pi_\text{old}}({s_t},\cdot)\big)} {Z^{\pi_\text{old}}(s_t)}\Big) \\
    =&\mathop{\arg\max}_{\pi\in\Pi}\mathbb{E}_{a_t\sim \pi}\left[Q^{\pi_\text{old}} (s_t,a_t) -\alpha \log\pi(a_t\vert s_t)\right]  {,} \nonumber
\end{align}
where $\mathrm{D}_{\mathrm{KL}}(\cdot)$ denotes the KL divergence, while the partition function $Z^{\pi_\text{old}}$ normalizes the distribution. The last equality in \eqref{eq:SAC policy KL} is due to that $Z^{\pi_\text{old}}$ has no contribution to gradient with respect to the new policy, it can thus be ignored. As unveiled in Appendix A, through the update rule of \eqref{eq:SAC policy KL}, $Q^{\pi_\text{new}}(s_t,a_t)\geq Q^{\pi_\text{old}}(s_t,a_t)$ is guaranteed for all $(s_t,a_t)\in \mathcal{S}\times\mathcal{A}$. 
Notably, the proof of soft policy improvement detailed in Appendix A additionally serves as a confirmation for policy improvement of regulated segment mixture, which will be elaborated upon later. 
\end{itemize}
Finally, with repeated application of soft policy evaluation and soft policy improvement, any policy $\pi\in\Pi$ will converge to the optimal policy $\pi^*$ such that $Q^{\pi^*}(s_t,a_t)\geq Q^{\pi}(s_t,a_t), \forall (s_t,a_t)\in \mathcal{S}\times\mathcal{A}$, and the proof can be found in \cite{sac1,sac2}. 

\section{System Model and Problem Formulation}\label{sec:system model and problem formulation}

\subsection{System Model}
We primarily consider a system consisting of $N$ agents empowered by the MASAC learning, which encompasses an independent local learning phase and a communication-assisted mixing phase. 
In the first phase, we use SAC algorithm for each IRL agent $i$, $i\in\{1,2,\cdots,N\}$. 
That is, agent $i$ senses partial status $s_t^{(i)}$ 
and has its local policy $\pi^{(i)}$ approximated by neural networks, and parameterized by $\theta\in\mathbb{R}^d$. 
It collects samples $\langle s_t^{(i)},a_t^{(i)},r_t^{(i)},s_{t+1}^{(i)} \rangle$ in replay buffer $\mathcal{D}^{(i)}$, and randomly samples a mini-batch $\Phi^{(i)}$ for local independent model updates\footnote{Hereafter, for simplicity of representation, we omit the superscript $(i)$ under cases where the mentioned procedure applies for any agent.}. 
Subsequently, in the second phase, each agent $i$ interacts with its one-hop neighbors $j\in\Omega_i$ within communication range, so as to reduce the behavioral localities of IRL and improve their cooperation efficiency. 

\subsubsection{Local Learning Phase}
Algorithmically, in the local learning phase with SAC, 
parameterized DNNs are used as approximators for policy and soft $Q$-function. 
Concretely, we alternate optimizing one network of policy $\pi$ 
parameterized by $\theta$ and two soft $Q$ networks parameterized by $\omega_{1}$ and $\omega_2$, respectively. 
Besides, there are also two target soft $Q$ networks parameterized by $\bar{\omega}_1$ and $\bar{\omega}_2$, obtained as an exponentially moving average of current $Q$ network weights $\omega_{1}$, $\omega_{2}$. 
Yet, only the minimum $Q$ value of the two soft $Q$-functions is used for the SGD and policy gradient. This setting of two soft $Q$-functions will speed up training while the use of target $Q$ can stabilize the learning \cite{sac1, sac2, duan2022distributional}. 
For training $\theta$ and $\omega\in{\omega_1,\omega_2}$, agent randomly samples a batch of transition tuples from the replay buffer $\mathcal{D}$ and performs SGD. 
The parameters of each soft $Q$ network $\omega_x, \forall x=1,2$ are updated through minimizing the soft Bellman residual error, that is,
\begin{align}
\vspace{-0.2em} 
    \label{eq:j_q}
    J_Q(\omega_x)=\frac{1}{2}\mathop{\mathbb{E}}\nolimits_{s_t,a_t \sim\mathcal{D}}\left[Q_{\omega_x}({s}_t,{a}_t)- \hat{Q}(s_t,a_t)\right]^2  {,}
    \vspace{-0.2em} 
\end{align}
where the target $\hat{Q}(s_t,a_t) = r_t+\gamma \mathop{\mathbb{E}}_{\substack{{s}_{t+1}\sim p(\cdot\vert s_t,a_t)\\{a_{t+1}\sim\pi_\theta}}}\!\left[
\mathop{\min}\limits_{x\in 1,2}
Q_{\bar{\omega}_x}(s_{t+1},a_{t+1})\!-\!\alpha \log\pi_\theta(a_{t+1}\vert s_{t+1})\right]$. 

Furthermore, the policy parameters of standard SAC can be learned according to \eqref{eq:SAC policy KL} by replacing $Q^{\pi_\text{old}}$ with current $Q$ function estimate as
\begin{equation}\label{eq:policy objective}
J_\pi(\theta)=\!\mathop{\mathbb{E}}\limits_{{s}_t\sim\mathcal{D}}
\!\big[\!
\mathop{\mathbb{E}}\limits_{{a}_t\sim\pi_\theta}
\!\big[\alpha\log\pi_\theta({a}_t|{s}_t)\!-\!\mathop{\min}_{\substack{x\in 1,2}}Q_{\omega_x}({s}_t,{a}_t)\big]\big] {.}
\vspace{-0.2em} 
\end{equation}

Besides, since the gradient estimation of \eqref{eq:policy objective} has to depend on the actions stochastically sampled from $\pi_\theta$, which leads to high gradient variance, the reparameterization trick \cite{sac1,sac2} is used to transform the action generation process into deterministic computation, allowing for efficient gradient-based training. 
Concretely, the random action $a_t$ can be expressed as a reparameterized variable 
\begin{align}
\vspace{-0.2em}
    a_t =f_\theta(s_t;\delta_t)=\tanh\big(\mu_\theta(s_t)+\delta_t\odot\sigma_\theta(s_t)\big) {,}
\vspace{-0.2em} 
\end{align}
where $\odot$ represents Hadamard product, 
$\delta_t$ is sampled from $\mathcal{N}(0,\mathbf{I}_{\dim \mathcal{A}})$, the mean $\mu_\theta$ and standard $\sigma_\theta$ are outputs from the policy network parameterized by $\theta$. This reparameterization enables the action $a_t$ to be a differentiable function of $\theta$, facilitating gradient descent methods. 

In addition to the soft $Q$-function and the policy, the temperature parameter $\alpha$ can be automatically updated by optimizing the following loss
\begin{equation}
J(\alpha)=\mathop{\mathbb{E}}\limits_{s_t\sim \mathcal{D}}\big[\mathop{\mathbb{E}}_{{a}_t\sim\pi_\theta}[-\alpha\log\pi_\theta({a}_t|{s}_t)-\alpha\bar{\mathcal{H}}]\big] {,}
    \label{eq:j_alpha}
    \vspace{-0.2em} 
\end{equation}
where $\bar{\mathcal{H}}$ is an entropy target with default value $-\dim{\mathcal{A}}$. Thus the policy can explore more in regions where the optimal action is uncertain, and remain more deterministic in states with a clear distinction between good and bad actions. 
Besides, since two-timescale updates, i.e, less frequent policy updates, 
usually result in higher quality policy updates, we integrate the delayed policy update mechanism employed in TD3 \cite{fujimoto2018addressing} into our framework's methodology. 
To this end, the policy, temperature and target $Q$ networks are updated with respect to soft $Q$ network every $e$ iterations. 

\noindent 

\subsubsection{Communication-Assisted Mixing Phase}
Subsequent to the phase of local learning, neighboring agents initiate periodic communication to enhance their collaboration in accomplishing complex tasks. 
The messages transmitted by agents are limited to policy parameters $\theta$ as in many works \cite{taik2022clustered, chellapandi2023federated, pacheco2024efficient}. Such an assumption is feasible as soft $Q$-functions have less impact on action selection than the policy in actor-critic algorithms.

Specifically, every $U$ times of policy updates, a communication round begins. 
Agent $i$ receives the policy parameters $\theta^{(j)}$ from neighboring agents $j\in \Omega_i$ via the 
D2D collaboration channel. 
Subsequently, agent $i$ could employ various methods to formulate a referential policy $\tilde{\pi}^{(i)}$, which is parameterized by $\tilde{\theta}^{(i)}\!=\!f({\theta}^{(1)},\cdots,{\theta}^{(j)},\cdots)$, that 
is, the parameters obtained from its neighbors $\forall j \in \Omega_i$. Afterward, agent $i$ directly mixes DNN parameters of the policy network, consistently with parallel distributed SGD methods as
\begin{equation}\label{eq:NN parameters mix}
\theta_\text{mix}^{(i)}=\theta^{(i)} +\zeta(\tilde{\theta}^{(i)}-\theta^{(i)}) {,}
\vspace{-0.1em}
\end{equation}
where $\zeta \in [0,1]$ is the mixture metric of DNN parameters. Taking model averaging in \cite{liu2022decentralized, xu2022trustable} as the example, $\tilde{\theta}^{(i)}$ is computed as $\tilde{\theta}^{(i)}\!=\frac{1}{|\Omega_i|}\!\sum\nolimits_{j\in\Omega_i}\!{\theta}^{(j)}$, and $\zeta\!=\!1\!-\!1/(|\Omega_i|+1)$ is further influenced by the number of neighbors involved. 
Then, for each agent $i$, $\theta^{(i)}$ should get aligned with mixed policy's parameters $\theta_{\text{mix}}^{(i)}$. 

\subsection{Problem Formulation}
This paper primarily targets the communication-assisted mixing phase. 
Instead of simply continuing to follow the idea of the direct average of the DNN parameters in FL, an effective parameter mixture method could better leverage the exchanged parameters to yield a superior target referential policy. This approach aims to improve RL policy performance during the communication-assisted mixing phase, resulting in a consistently higher value with respect to the maximum entropy objective, as described in \eqref{eq:MERL objective}.
However, it remains little investigated on the feasible means to mix the exchanged parameters (or their partial segments) and determine the proper mixture metric in \eqref{eq:NN parameters mix}, though it vitally affects both the communication overhead and learning performance. 
Therefore, by optimizing the mixture metric $\zeta$, we mainly focus on reducing the communication expenditure while maintaining acceptable cumulative rewards, that is,
\begin{align}
\vspace{-0.5em}
& \min_{\zeta} c(\upsilon,f) \nonumber \\
s.t.\quad 
& \sum\nolimits_t r_t(\Theta, \zeta) \geq r^{\text{thre}}  {,} \nonumber \\
& \Theta \leftarrow \{\theta_{\text{mix},k}^{(1)},\cdots, \theta_{\text{mix},k}^{(N)} \}  {,} \\
& \theta_{\text{mix},k}^{(i)}=\theta_k^{(i)} +\zeta(\tilde{\theta}_k^{(i)}-\theta_k^{(i)}) {,} &  \forall i \in \{1,\cdots, N\} {,} \nonumber \\
& \tilde{\theta}_k^{(i)} = f({\theta}_k^{(1)},\cdots,{\theta}_k^{(j)},\cdots) {,} & \forall k\text{ mod } U =0, j \in \Omega_i \nonumber {,} 
\end{align}
where $r^{\text{thre}}$ denotes the required minimum cumulative rewards, $v$ indicates the transmission bits of policy parameters, and $k$ is the index of policy iterations. 
Furthermore, $c(v,f)$ denotes the communication expenditure, governed by the utilization function $f$, which represents the method of utilizing the parameter or parameter segments obtained from different neighboring agent $j\in\Omega_i$ to construct a referential policy parameterized by $\tilde{\theta}$. 
The practical implementation of $f$ can be contingent on various factors, i.e., the underlying communications capability of agents to simultaneously receive and handle signals and the environmental conditions. 
After obtaining the referential policy parameter $\tilde{\theta}$, it is worthwhile to resort to a more comprehensive design of $\zeta$ to calibrate the communicating agents and contents as well as regulate the means to mix parameters or parameter segments, so as to provide a guarantee of performance improvement.

\section{Mixed Performance Improvement Bound Theorem of FRL Communication under MERL}
\label{sec:bound}

\begin{figure*}[tbp]
	\subfigcapskip = -10pt
	\begin{center}
		\subfigure[Independent local learning phase]
		{\includegraphics[scale =0.2]{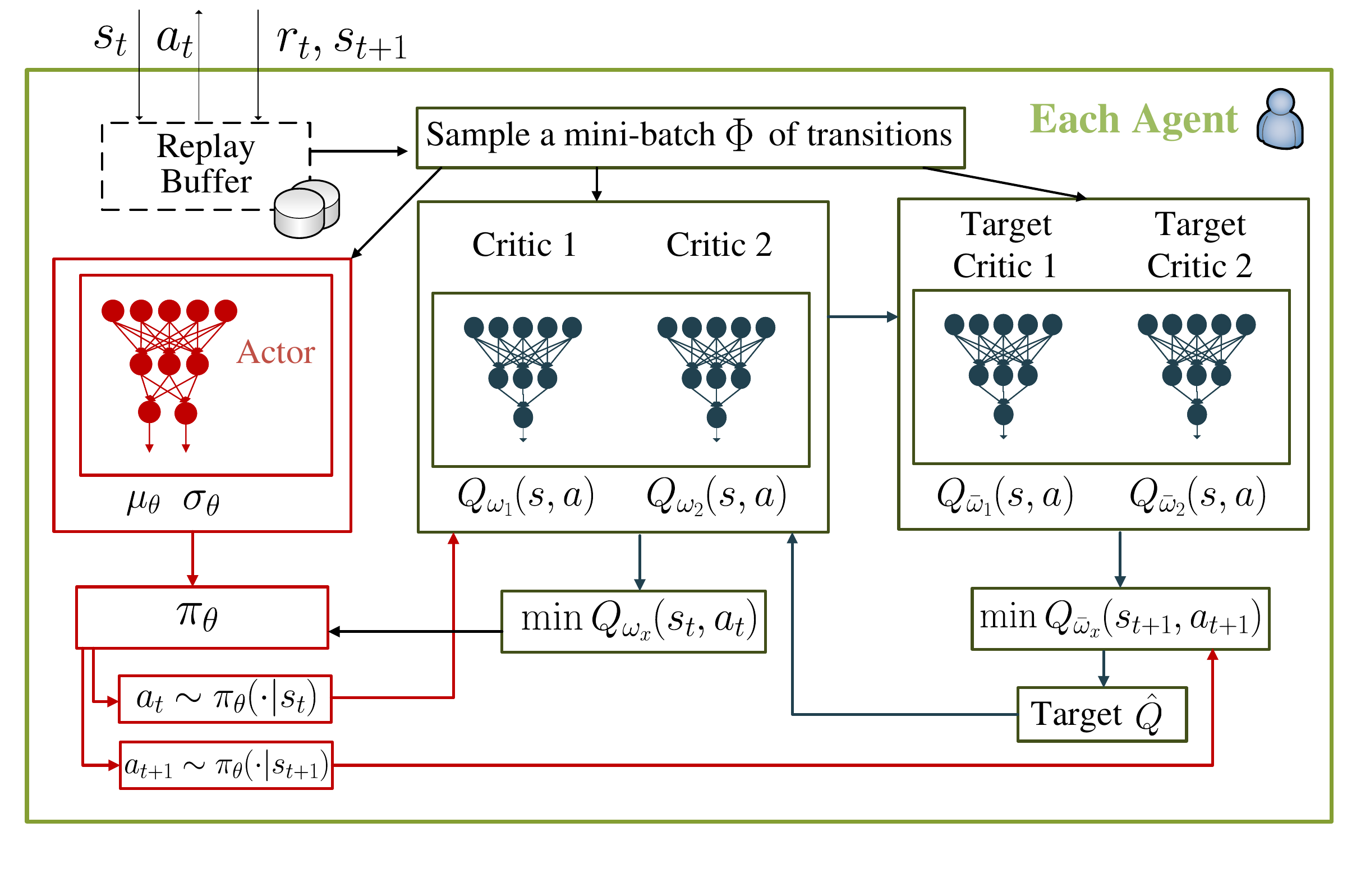}}
		\hspace{5pt}
		\subfigure[Communication-assisted mixing phase]
		{\includegraphics[scale =0.4]{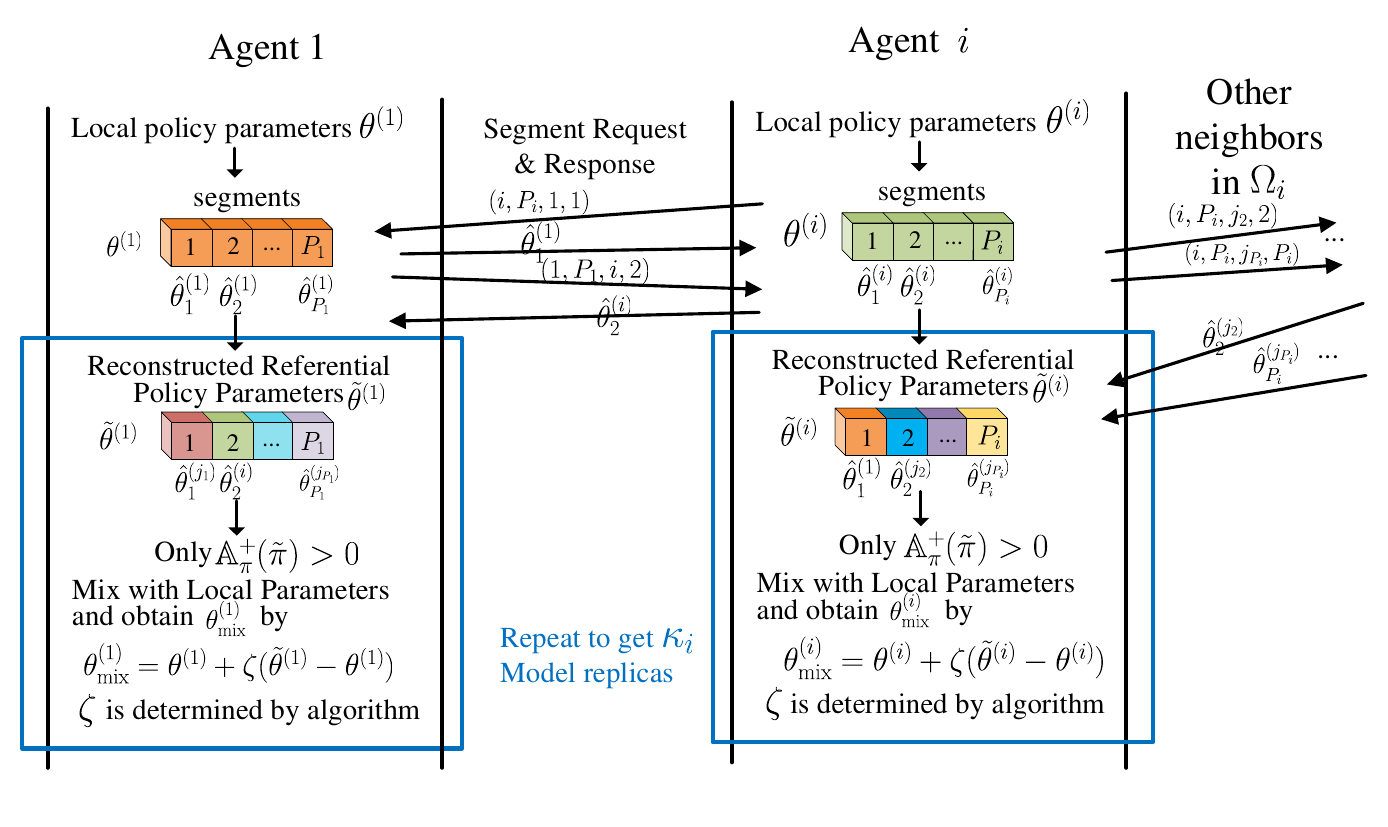}}
		\vspace{-0.3em}
		\caption{The illustration of RSM-MASAC implementation.}  
		\vspace{-0.6cm}
		\label{fig:framework}
	\end{center}
\end{figure*} 

Given the potential for significant variability in policy performance due to differences in training samples among multiple agents, as well as the staleness of iterations caused by the varying computing power of different agents, it is critical to recognize that not all referential policies - those amalgamated with the agent's own policy via the DNN parameter mixture approach detailed in \eqref{eq:NN parameters mix} - can contribute positively to agent's local learning process. 
More seriously, it may even degrade the learning performance sometimes \cite{xu2022trustable,kuba2022trust}.
Therefore, in order to ensure the policy improvement post-mixture, {a robust theoretical analysis method for evaluating the efficacy of the mixed policy distribution $\pi_\text{mix}$, approximated by DNN with parameter $\theta_\text{mix}$, is essential. }
This would regulate the selection of only those referential policies $\tilde{\pi}$, parameterized by $\tilde{\theta}$, that are beneficial, in the parameter mixture process. 

Drawing from this premise, based upon the conservative policy iteration as outlined in \cite{kakade2002approximately}, we employ a mixture update rule on policy distributions to find an approximately optimal policy. 
Our analysis here will focus specifically on calibrating the mixing means of policy distributions $\pi$ that exhibit a monotonic policy improvement property during the communication-assisted mixing phase. 
Notably, due to the non-linear transformation in DNN and possible applicability of softmax or restricted reparameterization \cite{agarwal2021theory}, 
    the mixture of policy distribution in \eqref{eq:update rule} is not directly equivalent to the mixture of their parameters, 
    especially when the parameter mixture takes the form in \eqref{eq:NN parameters mix} under distributed SGD. Hence, 
the detailed and more practical mixture implementation for the policy network parameter $\theta$, as in \eqref{eq:NN parameters mix}, will be derived from this section and more thoroughly presented in Section \ref{sec:NN mix section}. 

For any state $s$, we also define the mixed policy $\pi_\text{mix}$, which refers to a mixed distribution, as the linear combination of any referential policy $\tilde{\pi}$ and current policy $\pi$
\begin{equation}
    \label{eq:update rule}
    \pi_\text{mix}(a\vert s)=(1-\beta)\pi(a\vert s)+\beta\tilde{\pi}(a\vert s)  {,}
\end{equation}
where $\beta\in[0,1]$ is the weighting factor. 
The soft policy improvement, as outlined in Appendix A, suggests that this policy mixture can influence the ultimate performance regarding the objective stated in \eqref{eq:MERL objective}. 
Fortunately, we have the following new theorem on the performance gap associated with adopting these two different policies. 
\newtheorem{theorem}{Theorem}
\begin{theorem}{(Mixed Policy Improvement Bound)}
\label{theorem:mix policy improvement}
For any policy $\pi$ and $\tilde{\pi}$ adhering to \eqref{eq:update rule}, the improvement in policy performance after mixing can be measured by:
\begin{align}\label{eq:soft policy improvement inequality}
    \eta(\pi_\text{mix})-\eta(\pi)
    &\geq 
    \beta \mathop{\mathbb{E}}_{\substack{s\sim d_\pi\\
    a\sim \tilde{\pi}}}\left[A_\pi(s,a)+\alpha H(\tilde{\pi}(\cdot\vert s))\right] 
    - \frac{2\gamma\varepsilon \beta^2 }{(1-\gamma)^2} 
     \nonumber\\
     &\quad +\alpha  \mathop{\mathbb{E}}_{s\sim d_{\pi_\text{mix}}}\left[\mathrm{D}_{\mathrm{JS}}^\beta(\tilde{\pi}(\cdot\vert s)\Vert \pi(\cdot\vert s))\right]   {,}
\end{align}
where $\varepsilon \!:=\! \max_{s}\vert \mathop{\mathbb{E}}_{a\sim \tilde{\pi}}[A_\pi(s, a)\!+\!\alpha H(\tilde{\pi}(\cdot\vert s))]\vert$ represents the maximum advantage of $\tilde{\pi}$ relative to $\pi$, and 
$\mathrm{D}_{\mathrm{JS}}^\beta(p\Vert q)= \beta\sum p\log\frac{p}{\beta p + (1-\beta)q}+(1-\beta)\sum q\log\frac{q}{\beta p +(1-\beta)q}$ is the $\beta$-skew Jensen-Shannon(JS)-symmetrization of KL divergence \cite{nielsen2019jensen}, with two distributions $p$ and $q$. 
\end{theorem}
The proof of this theorem is given in Appendix B. 

\textbf{Remark:}
Notably, the proof effectively tackles the difficulties arising from the re-defined, soft state value function with the extra logarithmic term of the policy by utilizing entropy decomposition of $ H(\pi_{\text{mix}}(\cdot\vert s))$ in Lemma \ref{lemma:mix entropy} as well as several mathematical tricks. Hence, this sets the stage for theoretically evaluating the performance of the mixed policy before the mixture occurs. 
The mixed policy improvement bound in \eqref{eq:soft policy improvement inequality} implies that under the condition that the right-hand side of \eqref{eq:soft policy improvement inequality} is larger than zero, the mixed policy will assuredly lead to an improvement in the true expected objective $\eta$. 
Besides, from another point of view, any mixed policy with a guaranteed policy improvement in Theorem \ref{theorem:mix policy improvement} definitely satisfies the soft policy improvement as well, since the item $\mathbb{E}_{s_1}\left(\alpha H(\pi_\text{old}(\cdot\vert s_{1}))
\!+\!
\mathop{\mathbb{E}}_{\substack{a_{1}\sim \pi_\text{old}}} [Q^{\pi_\text{old}}(s_{1},a_{1})]\right)=\mathbb{E}_{s_1}[V^{\pi_\text{old}}(s_1)]$ in \eqref{eq:A.1} in the Appendix A (i.e., proof of Lemma \ref{lemma:appendix Soft policy improvement}) is less than $\mathbb{E}_{s_1} [V^{\pi_\text{mix}}(s_1)]$. 
In addition, when the temperature parameter $\alpha=0$, this inequality can reduce to the standard form of policy improvement in traditional RL 
\cite{kakade2002approximately, schulman2015trust, xu2022trustable}. 
Hence, Theorem \ref{theorem:mix policy improvement} derives a more general conclusion for both MERL and traditional RL. 

Furthermore, since for $\beta\in[0,1]$, $\mathrm{D}_{\text{JS}}^\beta$ is greater than zero \cite{nielsen2019jensen}, only the sign of the first two terms in \eqref{eq:soft policy improvement inequality} needs to be considered. 
By applying the re-defined policy advantage under MERL
\begin{align}\label{eq:definition of A}
    \mathbb{A}^{+}_\pi(\tilde{\pi}) := \mathop{\mathbb{E}}_{\substack{s\sim d_\pi, a\sim \tilde{\pi}}}\left[A_\pi(s,a)+\alpha H(\tilde{\pi}(\cdot\vert s))\right]  {,}
\end{align} 
we can therefore establish a tighter yet more tractable bound for policy improvement as
\begin{align}\label{eq:SAC cumulative reward bound}
    \eta(\pi_\text{mix})-\eta(\pi)
    &\geq 
    \beta \mathbb{A}^{+}_\pi(\tilde{\pi})
    - C \beta^2   {,}
\end{align}
where $C = \frac{2\varepsilon\gamma}{(1-\gamma)^2}$.
Thus, \eqref{eq:SAC cumulative reward bound} indicates that a mixed policy conforms to the principle of soft policy improvement, provided that the right side of \eqref{eq:SAC cumulative reward bound} yields a positive value. 
More specifically, if the policy advantage $\mathbb{A}^{+}_\pi(\tilde{\pi})$ is positive, an agent with policy $\pi$ can reap benefits by mixing its policy distribution with referential policy $\tilde{\pi}$. On the contrary, if this advantage is non-positive, policy improvement cannot be assured through the mixing of $\tilde{\pi}$ and $\pi$. In essence, \eqref{eq:SAC cumulative reward bound} serves as a criterion for selecting referential policies that ensure final performance improvement, and thus we have the following corollary. 
\begin{corollary}
    \label{cor:performance_improvement}
    To obtain guaranteed performance improvement, the mixture approach of policy distributions shall satisfy that 
    \begin{itemize}
        \item $\mathbb{A}^{+}_\pi(\tilde{\pi}) > 0$;
        \item $\beta \mathbb{A}^{+}_\pi(\tilde{\pi})
        - C \beta^2 > 0$.
    \end{itemize}
\end{corollary}

\section{MASAC with Regulated Segment Mixture}\label{sec:MASAC design} 
In this section, 
as shown in Fig. \ref{fig:framework}, 
we present the design of RSM-MASAC, which reduces the communication overhead while incurring little sacrifice to the learning performance. 
\subsection{Algorithm Design}\label{4B1}
Consistent with the SAC setting as in Section \ref{sec:system model and problem formulation}, agents in RSM-MASAC undergo the same local iteration process. 
Meanwhile, for the communication-assisted mixing phase, we will elaborate on the design details of RSM-MASAC, {including the segment request \& response and policy parameter mixture with theory-established performance improvement. }

\subsubsection{Segment Request \& Response}
Inspired by segmented pulling synchronization mechanism in DFL \cite{hegedHus2019gossip}, we develop and perform a segment request \& response procedure. The proposed approach divides transmission of the policy parameters into segments, and each agent selectively requests various segments of policy parameters from different neighbors simultaneously through D2D communication, thereby facilitating the construction of a reconstructed referential policy for subsequent aggregation while also effectively balancing the load of communication costs and optimizing bandwidth usage. 
Specifically, for every communication round, each agent $i$ breaks its policy parameters $\theta^{(i)}$ into $P_i$ ($P_i =\min\{P, \vert \Omega_i \vert \}$, according to the default segmentation granularity $P$ and the number of neighbors $\vert \Omega_i \vert$ in current time) non-overlapping segments $\hat{\theta}^{(i)}_1,\hat{\theta}^{
(i)}_2,\cdots,\hat{\theta}^{(i)}_{P_i}$ as
\vspace{-0.1em}
\begin{equation}
    \theta^{(i)} =(\hat{\theta}^{(i)}_1,\hat{\theta}^{(i)}_2,\cdots,\hat{\theta}^{(i)}_{P_i})   {.}
\vspace{-0.1em}
\end{equation}
Significantly, the available segmentation strategies are diverse and include, but are not limited to, dividing the policy parameters by splitting DNN layers \cite{barbieri2023layer}, modular approach \cite{barbieri2022decentralized}, or segmenting according to the parameter size \cite{hegedHus2019gossip}. 
Without loss of generality, taking IoV as an application example, there exist rather diverse V2V communication protocols and technologies such as DSRC (based on IEEE 802.11p and its subsequent IEEE 802.11bd), NR-V2X (as specified in 3GPP Release 16 and 17), which support Multiple Input Multiple Output (MIMO) technology and advanced signal processing techniques, such as spatial multiplexing and beamforming. 
Therefore, the communication capability depends on multiple factors ranging from physical layer configurations, such as antenna array setup and channel conditions, to upper layer protocols, shaping the theoretical upper limit for segmentation granularity $P_{\max}$. 
Besides, shorter distances between agents (e.g., vehicles) and better channel conditions generally enhance the capacity and reliability of the system. 
Still, the segmentation also can be dynamic, with larger $P_{\max}$ in environments with higher agent density and smaller $P_{\max}$ in cases where fewer agents are within communication range.

To clarify this process, we use the most intuitive uniform parameter partition. 
For each segment $p=1,\cdots, P_i$, agent $i$ randomly selects a target agent (without replacement) from its neighbors (i.e., $j_p \in \Omega_i$) to send segment request $(i,P_i,j_p,p)$, which indicates the agent $i$ who initiates the request, its total segment number $P_i$, as well as the requested segment $p$ from the target agent $j_p$. Upon receiving the request, the agent $j_p$ will break its own policy parameters $\theta^{(j_p)}$ into $P_i$ segments and return the corresponding requested segment $\hat\theta^{(j_p)}_p$ according to the identifier $p$. 
Then, agent $i$ could reconstruct a referential policy based on all of the fetched segments, that is, 
\begin{equation}\label{eq:reconstruct policy}
    \tilde{\theta}^{(i)} = (\hat{\theta}^{(j_1)}_1,\hat{\theta}_2^{(j_2)},\cdots,\hat{\theta}_{P_i}^{(j_{P_i})})  {.}
\vspace{-0.1cm}
\end{equation}

In fact, this segmented transmission approach can be executed in parallel, thereby optimizing the utilization of available bandwidth. Instead of being confined to a single link, the traffic is distributed across $P_i$ links, enhancing the overall data transfer efficiency. 
Besides, in order to further accelerate the propagation and ensure the model quality, we can construct multiple model replicas in RSM-MASAC. 
That is, the process of segment request $\&$ response can be repeated $P_i\times \kappa_i$ times, reconstructing $\kappa_i =\min\{\kappa, \vert \Omega_i \vert \}$ reconstructed referential policies in one communication round.

\begin{algorithm}[tbp]\small
  \caption{The RSM-MASAC Algorithm.}
  \label{my alg}
  \begin{algorithmic}[1]
  \STATE{Initialize network parameters $\theta^{(i)}$, $\omega_{1}^{(i)}$, $\omega_{2}^{(i)}$, $i=1,2,\cdots N$.}
  \STATE{Initialize target network parameters $\bar{\omega}_1^{(i)}\leftarrow \omega_1^{(i)}$, $\bar{\omega}_2^{(i)}\leftarrow \omega_2^{(i)}$.}
  \STATE{Initialize learning rate $\eta_\pi$, $\eta_Q$, $\eta_\alpha$, temperature $\alpha$, communication interval $U$, segmentation granularity $P$, predefined replicas $\kappa$.}
  \STATE{Initialize iteration index $k\gets 0$ and $ counter \gets 0$.}
    \FOR{each epoch}
        \FOR{$t\leftarrow 1$ to $T$}
        \STATE {\textbf{Each agent $i$ executes:}}
        \STATE{\textbf{* Independent local learning phase }}
                    \STATE{Select an action $a_t^{(i)}$ with respect to $s_t^{(i)}$ according to the current policy $\pi^{(i)}$ parameterized by $\theta^{(i)}$.}
                    \STATE{Observe reward $r_t^{(i)}$ and next state $s_{t+1}^{(i)}$.}
                    \STATE{Save the new transition in replay buffer:
                    $\mathcal{D}^{(i)}\gets\mathcal{D}^{(i)}\cup\langle s_t^{(i)},a_t^{(i)},r_t^{(i)},s_{t+1}^{(i)} \rangle$.}
                    \STATE{Sample a mini-batch $\Phi^{(i)}\sim \mathcal{D}^{(i)}$.}
                    \STATE{Update soft $Q$ function $\omega_x^{(i)} \gets \omega_x^{(i)}-\eta_Q \nabla J_{Q}(\omega_x^{(i)})$ by \eqref{eq:j_q}, for $\forall x={1,2}$.}
                    \IF{$ counter \bmod d = 0$ \textbf{or} $t=T$}
                        \STATE{Update policy $\vspace{-0.5em}\theta_{k+1}^{(i)} \leftarrow \theta_k^{(i)}-\eta_{\pi} \nabla J_\pi (\theta_k^{(i)})$ by \eqref{eq:policy objective}.}
                        \vspace{5pt}
                        \STATE{Adjust temperature $\alpha^{(i)} \gets \alpha^{(i)}-\eta_{\alpha} \nabla J(\alpha^{(i)})$ by \eqref{eq:j_alpha}.}
                        \vspace{1pt}
                        \STATE{Update target networks $\bar{\omega}_x^{(i)}\gets \varrho\omega_x^{(i)}+(1-\varrho)\bar{\omega}_x^{(i)}$, for $\forall x=1,2$.}
                        \STATE{$k\gets k+1$.}
                        
                    \ENDIF
                     \STATE{$counter\gets counter + 1$.}
                    
                    \STATE{\textbf{* Communication-assisted mixing phase }}
                    \IF{$ k \bmod  U = 0$}
                        \STATE{Update policy $\theta_k^{(i)}\leftarrow$ \textbf{CommMix}($\theta_k^{(i)}$, $i$, $\Omega_i$, $P$, $\kappa$) according to Algorithm \ref{al:CommMix}.}
                    \ENDIF                      
        \ENDFOR
        \ENDFOR
  \end{algorithmic}
\end{algorithm}
   
\begin{algorithm}[tbp]\small
  \caption{The \textbf{CommMix} Function in Algorithm \ref{my alg}.}
  \label{al:CommMix}   
  \begin{algorithmic}[1]
        \STATE{\textbf{Input:} $\theta$, $i$, $\Omega_i$, $P$, $\kappa$.}
        \STATE{$P_i = \min\{P, \vert \Omega_i \vert \}$ and $\kappa_i =\min\{\kappa, \vert \Omega_i \vert \}$.}
        \FOR{each replica $1,2,\cdots,\kappa_i$}

            \STATE{Send $P_i$ pulling request $(i,P_i,j_p,p)$ to nearby collaborators in $\Omega_i $, and receive $\hat{\theta}^{(j_p)}_p$ to reconstruct $\tilde{\theta}$ as \eqref{eq:reconstruct policy}.}
            \STATE{Select $M$ samples from the replay buffer $\mathcal{D}^{(i)}$.}
            \STATE{Estimate $\mathbb{A}^{+}_{\pi}(\tilde{\pi}) $ according to \eqref{eq:estimate policy advantage}.}
            \IF{$\mathbb{A}^{+}_\pi(\tilde{\pi}) >0$}
                \STATE{Evaluate $F(\theta)$ according to \eqref{eq:estimate FIM}.}
                \STATE{Get the upper bound of $\zeta$ according to Theorem \ref{theorem: parameter mix}.}
                \STATE{Make the mixture metric $\zeta$ less than the calculated upper bound, and update $\theta_\text{mix}^{(i)}$ by \eqref{eq:NN parameters mix}.}
            \ENDIF
        \ENDFOR
        \STATE{\textbf{Output:} $\theta_\text{mix}^{(i)}$.}
  \end{algorithmic}
\end{algorithm}

\subsubsection{Policy Parameter Mixture with Theory-Established Performance Improvement}\ 
\label{sec:NN mix section}

Consistent with TRPO \cite{schulman2015trust}, we introduce KL divergence to replace $\beta$ by setting  
$\beta:=\sqrt{\mathrm{D}_{\mathrm{KL}}^{\max}(\pi\Vert\pi_{\text{mix}})}$, where $\mathrm{D}_{\mathrm{KL}}^{\max}(\pi\Vert\pi_{\text{mix}}) = \max_s \mathrm{D}_{\mathrm{KL}}(\pi(\cdot\vert s)\Vert\pi_{\text{mix}}(\cdot\vert s))$. 
Thus, the second condition in Corollary \ref{cor:performance_improvement} is equivalent to 
\begin{align}
\vspace{-0.2em}
	\sqrt{\mathrm{D}_{\mathrm{KL}}^{\max}(\pi\Vert\pi_{\text{mix}})}<\frac{\mathbb{A}^{+}_\pi(\tilde{\pi})}{C}.\label{eq:condition}
 \vspace{-0.2em}
\end{align} 

Nevertheless, it hinges on the computation-costly KL divergence to quantify the difference between probability distributions. 
Fortunately, since for a small change in the policy parameters, the KL divergence between the original policy and the updated policy can be approximated using a second-order Taylor expansion, wherein FIM serves as the coefficient matrix for the quadratic term. This provides a tractable way to assess the impact of parameter changes on the policy. 
Therefore, 
we utilize FIM, delineated in context of natural policy gradients by \cite{kakade2001natural}, as a mapping mechanism to revise the impact of certain changes in policy parameter space on probability distribution space. 
Then in the following theorem, we can get the easier-to-follow, trustable upper bound for the mixture metric of policy DNN parameters.
\begin{theorem}(Guaranteed Policy Improvement via Parameter Mixing)
	\label{theorem: parameter mix}
    With any referential policy parameters $\tilde{\theta}$, an agent with current policy parameters $\theta$ can improve the true objective $\eta$ as in \eqref{eq:MERL objective} through updating $\theta$ to mixed policy parameters $\theta_\text{mix}$ in accordance with \eqref{eq:NN parameters mix}, provided it fulfills the following two conditions:
    \begin{itemize}
        \item $\mathbb{A}^{+}_\pi(\tilde{\pi})> 0$;
        \item $0<\zeta<\left[\frac{2\mathbb{A}^{+}_\pi(\tilde{\pi})}{C [(\tilde{\theta}-\theta)^{\intercal}  F(\theta) (\tilde{\theta}-\theta)]}\right]^\frac{1}{2}$.
    \end{itemize}
    where $F(\theta)$ is the FIM of policy $\pi$ parameterized by $\theta$. 
\end{theorem}
\begin{proof}
Recalling the definition of $\beta$ and the mixture approach of $\theta_\text{mix}$ in \eqref{eq:NN parameters mix}, as well the change of policy parameters (i.e., $\Delta\theta=\zeta(\tilde{\theta}-\theta)$), for any state $s$, the KL divergence between the current policy and the mixed policy  can be expressed  by performing a second-order Taylor expansion of the KL divergence at the point $\theta$ in parameter space as 
\vspace{-0.2em}
\begin{align}
 \vspace{-3em}
	\mathrm{D}_{\mathrm{KL}}&(\pi_\theta(\cdot\vert s)\Vert \pi_{\theta+\Delta\theta}(\cdot\vert s))\nonumber\\
	=& \sum_{a\in\mathcal{A}}\pi_\theta(a\vert s)\log \frac{\pi_\theta(a\vert s)}{\pi_{\theta+\Delta\theta}(a\vert s)}\nonumber\\
   \approx&
    \mathrm{D}_{\mathrm{KL}}(\pi_\theta(\cdot\vert s)\Vert \pi_{\theta}(\cdot\vert s)) 
   -\mathop{\mathbb{E}}\limits_{a\sim\pi_\theta}\left[\frac{\partial\log\pi_\theta(a\vert s)}{\partial \theta}\right]^{\intercal}\!\!\Delta\theta \nonumber\\
    &\quad + 
    \frac{1}{2}\Delta\theta^{\intercal}  F(\theta) \Delta\theta\nonumber\\
    \stackrel{(a)}{=} & \frac{1}{2}\zeta^2 (\tilde{\theta}-\theta)^{\intercal}  F(\theta)(\tilde{\theta}-\theta)  {,} \nonumber
    \vspace{-2em}
\end{align}
\noindent where the equality $(a)$ comes from the fact that by definition, $ \mathrm{D}_{\mathrm{KL}}(\pi_\theta(\cdot\vert s)\Vert \pi_{\theta}(\cdot\vert s)) = 0$, while 
$
\mathop{\mathbb{E}}\nolimits_{a\sim\pi_\theta}\left[\frac{\partial\log\pi_\theta(a\vert s)}{\partial \theta}\right] 
	 = \sum_{a\in\mathcal{A}} \pi_\theta(a\vert s)\frac{\partial\log\pi_\theta(a\vert s)}{\partial \theta}
	 = \sum_{a\in\mathcal{A}} \frac{\partial\pi_\theta(a\vert s)}{\partial \theta}
	 =0
$. 
Notably, the FIM $F(\theta)$ takes the expectation for all possible states, which reflects the average sensitivity of the whole state space rather than a particular state, and can be calculated as
\begin{equation}
 \vspace{-0.5em}
    F(\theta) = \mathop{\mathbb{E}}_{\substack{s\sim d_\pi \\ a\sim \pi_\theta}}\left[\left(\frac{\partial\log\pi_\theta(a\vert s)}{\partial \theta}\right)\!
    \left(\frac{\partial\log\pi_\theta(a\vert s)}{\partial \theta}\right)^{\intercal} \right]  {.}\nonumber
\end{equation}

Finally, combining Corollary \ref{cor:performance_improvement} and \eqref{eq:condition}, we have the theorem.
\end{proof}

\textbf{Remark:}  
With Theorem \ref{theorem:mix policy improvement} and Theorem \ref{theorem: parameter mix}, we can anticipate the extent of policy performance changes resulting from policy DNN parameters mixture during the communication-assisted mixing phase. 
This prediction is based solely on the local curvature of the policy space provided by FIM, thus avoiding the cumbersome computations to derive the full KL divergence of two distributions for every state. 
Besides, given that $F(\theta)$ is a positive definite matrix, a positive sign of the policy advantage $\mathbb{A}^{+}_\pi(\tilde{\pi})$ implies an increase in the DNN parameters' mixture metric $\zeta$ corresponding to the rise in $\mathbb{A}^{+}_\pi(\tilde{\pi})$. That is, rather than simple averaging, agents can achieve guaranteed soft policy improvement after mixing by learning more effectively from referential policies with an elevated $\zeta$. 

In practice, to evaluate $\mathbb{A}^{+}_\pi(\tilde{\pi})$ and $ F(\theta)$, the expectation can be estimated by the Monte Carlo method, 
approximating the global average by the states and actions under the policy. Meanwhile, the importance sampling estimator is also adopted to use the off-policy data in the replay buffer for the policy advantage estimation, where $\pi_{t}$ typically denotes the action sampling policy at time step $t$. As outlined in Lemma \ref{lemma:estimate of policy advantage} in Appendix C, we have
\begin{align}
\mathbb{A}^{+}_\pi(\tilde{\pi}) 
&\approx 
\mathbb{E}_{s_t,a_t\sim\mathcal{D}}
\bigg[
\left(\frac{\tilde{\pi}_{\tilde{\theta}}(a_t\vert s_t)-\pi_\theta(a_t\vert s_t)}{\pi_t(a_t\vert s_t)}\right)
 \!\mathop{\min}_{\substack{x\in 1,2}}Q_{\omega_x}({s}_t,{a}_t)\nonumber\\
& \qquad + \alpha[H(\tilde{\pi}_{\tilde{\theta}}(\cdot\vert s_t)) - H({\pi}_{{\theta}}(\cdot\vert s_t)) ]\bigg]  {,} \label{eq:estimate policy advantage}
\end{align}
and
\begin{equation}
F(\theta) \!\approx\!\mathop{\mathbb{E}}\limits_{s_t\sim\mathcal{D}}\!\bigg[\!\mathop{\mathbb{E}}\limits_{a_t\sim\pi_\theta}\!
\bigg[\!
\left(\frac{\partial\log\pi_\theta(a\vert s)}{\partial \theta}\right)\!\left(\frac{\partial\log\pi_\theta(a\vert s)}{\partial \theta}\right)^{\intercal}
\!\bigg]\bigg]  {.}\label{eq:estimate FIM}
\end{equation} 

Finally, we summarize the details of RSM-MASAC in Algorithm \ref{my alg}. 
RSM-MASAC employs a theory-guided metric for policy parameter mixture, which takes into account the potential influence of parameters mixture on soft policy improvement under MERL — a factor commonly overlooked in parallel distributed SGD methodologies. In particular, the mixing process between any two agents is initiated solely when there is a positive policy advantage. Adhering to Theorem \ref{theorem: parameter mix}, the mixture metric is set marginally below its computed upper limit, ensuring both policy improvement and convergence.

\subsection{{Communication Cost Analysis}}\label{4B2}
Since the pulling request does not contain any actual data, its cost in the analysis can be ignored, and we only consider the policy parameters transmitted among agents to analyze the communication efficiency of RSM-MASAC. 

Regarding the communication overhead for each segment request, RSM-MASAC incurs a maximization data transmission cost of $\upsilon_p=\upsilon/P$ through D2D communications. 
Following the same settings as in \cite{barbieri2022decentralized}, we provide illustrative IoV examples to demonstrate the overall reduction in communication overhead. 
Specifically, the Collective Perception Message (CPM) defined in Collective Perception Service (CPS) of ETSI TR 103 562 \cite{ETSI} can encapsulate DNN parameters into the Perceived Object Containers (POCs) and propagate them among cooperating DRL agents with each message carrying a payload of $4,480$ bytes, serving as segment responses in our algorithm. 
And the DNN parameters $\theta\in\mathbb{R}^d$ 's transmission size $\upsilon$ can be regarded as $32d$ bits, commonly assumed in \cite{reisizadeh2020fedpaq, barbieri2022decentralized, barbieri2022communication}. 
Consequently, the total communication overhead for reconstructing maximal $\kappa$ referential policies per agent in each round $c(\upsilon,f)=\frac{N\times \upsilon\times\kappa}{8 \times1024^3}$ (GB), as well as $\frac{N\times \upsilon\times\kappa}{8\times4480}$ message numbers, which is $(N-1)/\kappa$ times less than that in a fully connected communication setup \cite{xu2022trustable}. Additionally, communication occurs at a periodic interval of $U$, allowing for further reduction in communication overhead by decreasing communication frequency. Moreover, by simultaneously requesting $P_i$ agents in parallel, RSM-MASAC benefits from the sufficient use of the bandwidth and enhances the capability to overcome possible channel degradation.

\subsection{Complexity Analysis}
\label{sec:complexity}
	For each referential policy, the calculation of policy advantage $\mathbb{A}^{+}_\pi(\tilde{\pi})$ is contingent with sample size $M$, policy parameter size $\vert d \vert$ and $Q$ network parameter size $\vert d_Q \vert$, with complexity $\mathcal{O}(M(\vert d \vert+\vert d_Q \vert))$. 
    However, the mixture of policy parameters only proceeds if condition $\mathbb{A}^{+}_\pi(\tilde{\pi})>0$ is met during the communication-assisted mixing phase, with the computation complexity primarily depending on the calculation of FIM. 
    Specifically, with sampling approximated method using first-order gradients in \eqref{eq:estimate FIM}, the computation complexity of FIM is $\mathcal{O}(M\vert d \vert^2)$ and the associated memory complexity is $\mathcal{O}( \vert d \vert^2)$. 
	However, the actual computation of the FIM is not restricted to sampling-based method alone, there has been extensive research on FIM approximation methods, including Diagonal Approximation \cite{amari1998natural,george2018fast} and Low-Rank Approximation \cite{martens2010deep}, among others, to further reduce to the complexity.

\section{Experimental Results and Discussions}\label{sec:simulations}
In this section, we validate the effectiveness of our proposed algorithm for the speed control of Connected Automated Vehicles (CAVs) in IoV, highlighting its superiority compared to other methods.

\subsection{Experimental Settings}\label{4B}
We implement two simulation scenarios on Flow \cite{wu2017flow, vinitsky2018benchmarks}, which is a traffic control benchmarking framework for mixed autonomy traffic. As illustrated in Fig. \ref{fig:benchmark}, the common urban traffic intersection 
scenario ``Figure 8" and highway scenario ``Merge" are selected, with the main system settings described in Table \ref{tab:system setting}. 
\begin{figure}[tbp]
	\subfigcapskip = -1pt
	\begin{center}
		\subfigure[``Figure 8"]
		{\includegraphics[scale =0.05]{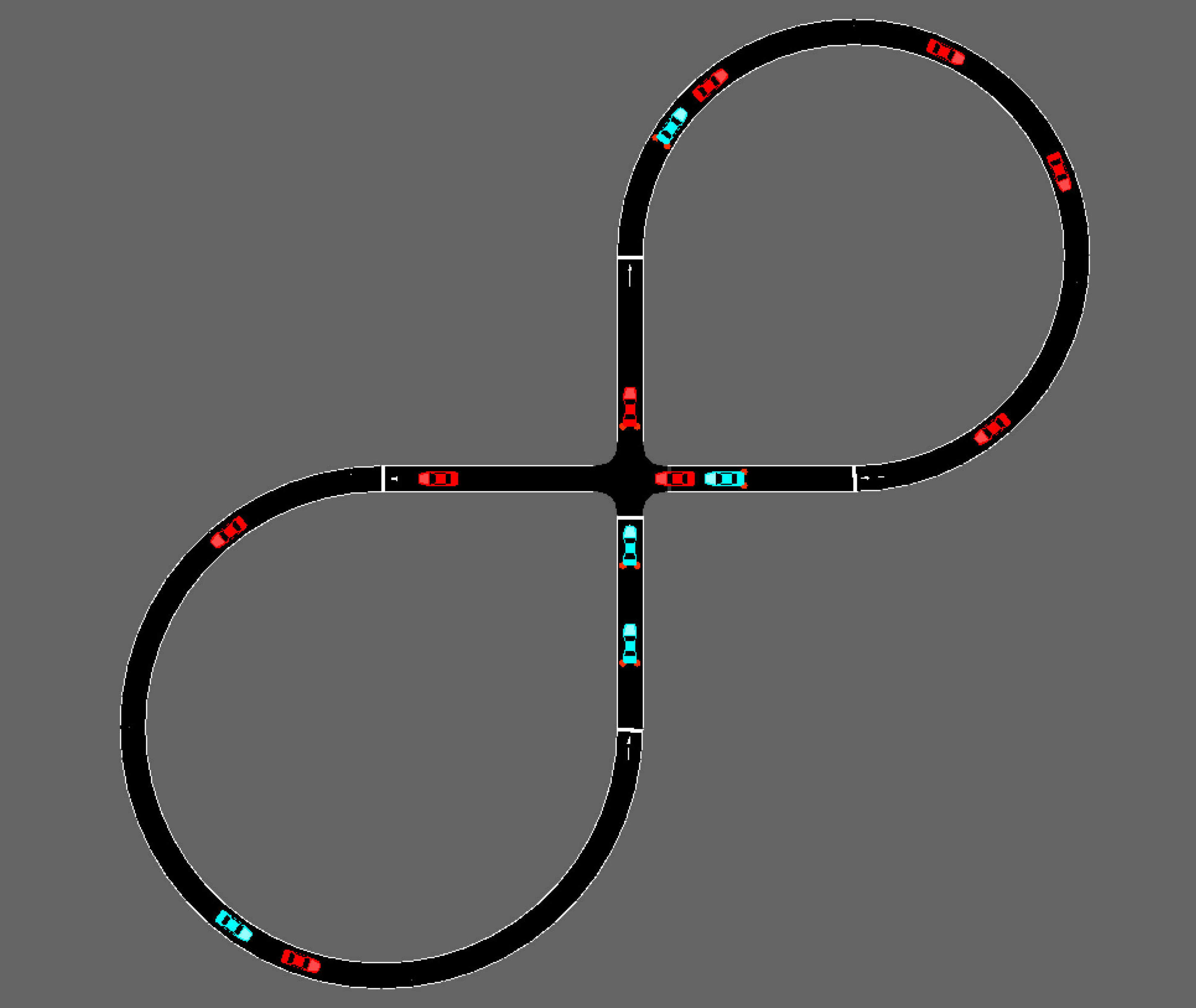}}
		\hspace{5pt}
		\subfigure[``Merge"]
		{\includegraphics[scale =0.52]{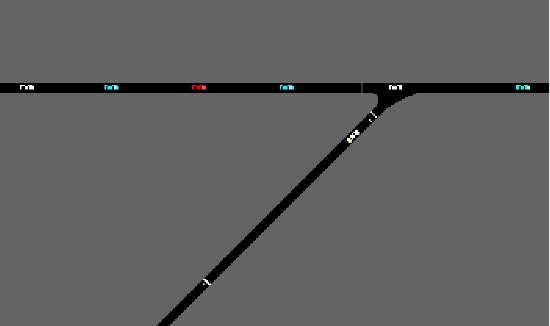}}
		\caption{The two scenarios for simulations on Flow. The red vehicles are DRL-driven CAVs, while the blue vehicles are the HDVs observed by the DRL-driven CAVs, and the white vehicles are the HDVs that are not observed in the state space.}  
		\label{fig:benchmark}
	\end{center}
 \vspace{-1em}
\end{figure} 

\begin{table}[tbp]
	\centering
	\caption{Setting of System Parameters in two scenarios.}
	\label{tab:system setting}
	\begin{tabular}{c|c|c}
		\toprule
		Parameters  Definition &Figure 8& Merge\\ \midrule
		Number of DRL agents $N$& $9$ & $13$\\
		Total time-steps per epoch $T$ & $1,500$ & $750$ \\
		Number of epochs &$300$&$260$\\
		Range of acceleration ($m/s^2$)& $[-3, 3]$ & $[-1.5, 1.5]$ \\
		Desired velocity per vehicle ($m/s$) & $20$ & $20$\\
		Speed limit per vehicle ($m/s$) & $30$ & $30$\\
		Length per time-step (s)& $0.1$ & $0.1$\\
            Maximum of vehicles per hour & - & $2,300$\\
		\bottomrule
	\end{tabular}
	
\end{table}
\begin{itemize}
\item ``Figure 8": $14$ vehicles navigate a one-way lane shapes like a figure ``8'', including $5$ emulated Human Driven Vehicles (HDV) controlled by Simulation of Urban MObility (SUMO) with Intelligent Driver Model (IDM) \cite{treiber2000congested}, and $9$ IRL-controlled CAVs maintaining dedicated links to update their parameters through the V2V channel. 
At the lane's intersection, each CAV adjusts its acceleration to traverse efficiently, aiming to boost the traffic flow's average speed. 

\item ``Merge": A highway on-ramp merging scenario with vehicle flow at $2,300$ per hour, including a maximum of $2, 200$ vehicles on the main road and $100$ vehicles on the ramp. 
Within each epoch, $13$ vehicles are randomly chosen to instantiate the DRL-based controllers as they sequentially enter, aiming to manage collision avoidance and congestion at merge points. The simulation settings closely align with those of the first scenario. 
\end{itemize}
Both two scenarios are modified to assign the limited partial observation of the global environment as the state of each CAV, including the position and speed of its own, the vehicle ahead and behind. 
Only CAVs can execute the V2V end-to-end communication, and the connectivity of V2V links at time $t$ depends on the CAVs' position and communication range, which are both extracted using the TraCI simulator in our experiments. 
A communication range of $90$ m is adopted in interactions ``Figure 8" and $400$ m in highway ``Merge". 
This is according to the conclusion in \cite{yang2023dynamic} that communication range at intersections will be extremely reduced compared with the conventional scenarios \cite{thota2019v2v}, and relevant service requirements\cite{3GPPServicerequirements}. 
{Unless otherwise stated, our experiments are based on the common MARL V2V lossless ideal communication premise \cite{han2023multi, chen2024communication, shi2023deep, qu2024model}.}
Meanwhile, each CAV's action is a continuous variable representing speed acceleration or deceleration. It is sampled from the outputs of the policy network, which has three fully connected layers with $256$ hidden units each and ReLU activations. Actions are squashed using a tanh function to fall within $[-1,1]$ and then scaled by the acceleration bounds of the scenarios. 

\begin{table}[tbp]
	\centering
	\caption{Hyper-parameters.}
	\label{tab:parameters}
	\begin{tabular}{l|c|c}
		\toprule
		Hyper-parameters  & Symbol & Value\\ \midrule
		Replay buffer size & $\vert\mathcal{D}\vert$ & $10^5$ \\
		Batch size & $\Phi$ & $256$ \\
		Number of samples to evaluate $\mathbb{A}^{+}_{\pi}(\tilde{\pi})$ &$M$ & $50 $ \\
		Learning rate of actor network &$\eta_\pi$ & $4\times 10^{-5}$ \\
		Learning rate of critic network &$\eta_Q$ & $3\times 10^{-4}$ \\
		Learning rate of temperature parameter &$\eta_\alpha$ & $3\times 10^{-4}$ \\
		Discount factor &$\gamma$ & $0.99$ \\ 
		Target smoothing coefficient &$\varrho$ & $10^{-3}$ \\ 
		Delayed policy update intervals &$e$ & $10$ \\ 
		Communication intervals & $U$ & $8$ \\
		Segmentation granularity & $P$ & 4\\
		Predefined replicas  & $\kappa$ & 3\\
		\bottomrule
	\end{tabular}
\end{table}

\begin{figure}[tbp]
	\subfigcapskip = -1pt
	\begin{center}
		\subfigure[{Runtime on CPU}]
		{\includegraphics[scale =0.33]{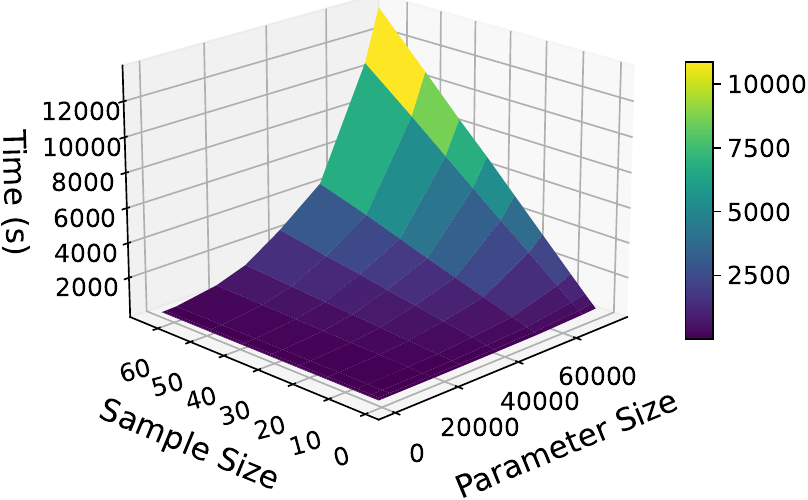}} 
		\subfigure[{Runtime on GPU}]
		{\includegraphics[scale =0.33]{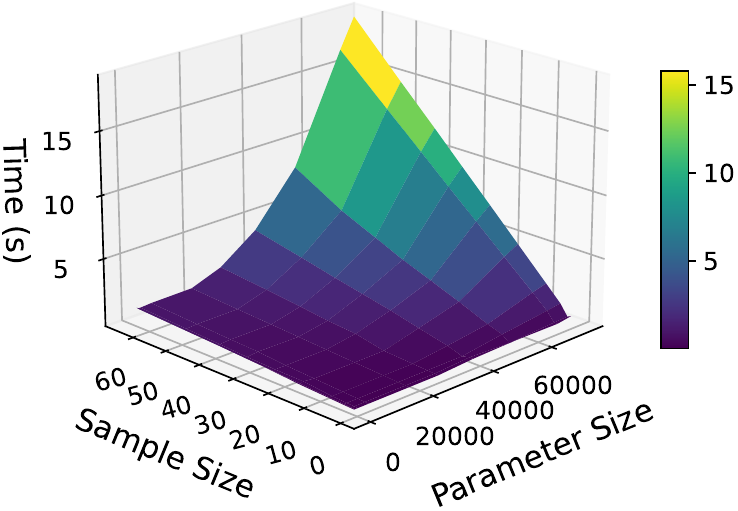}}
		\caption{{The runtime of estimating mixture metric under different sample and parameter sizes.}} 
        \label{fig:time}
	\end{center}
 \vspace{-1.5em}
\end{figure} 

\begin{figure*}[tbp]
	\subfigcapskip = -1pt
	\begin{center}
		\subfigure[{``Figure 8"}]
		{\includegraphics[scale =0.38]{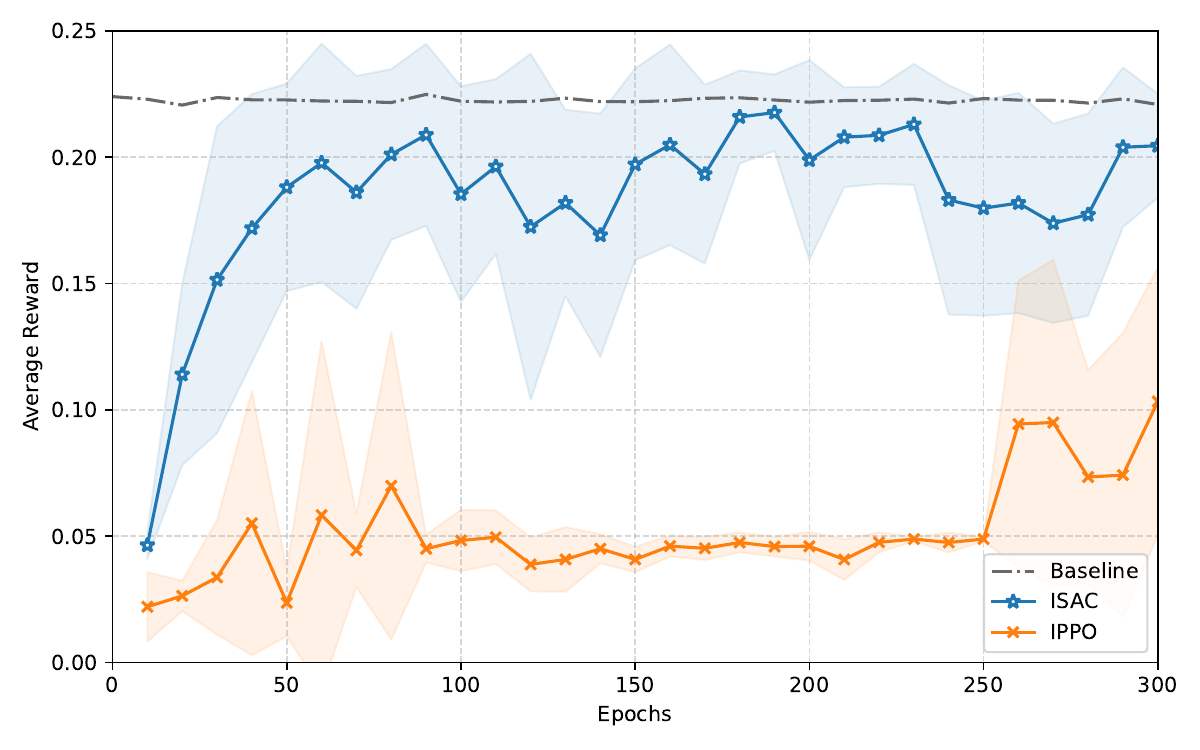}}
         \hspace{0.8cm}
		\subfigure[{``Merge"}]
		{\includegraphics[scale =0.379]{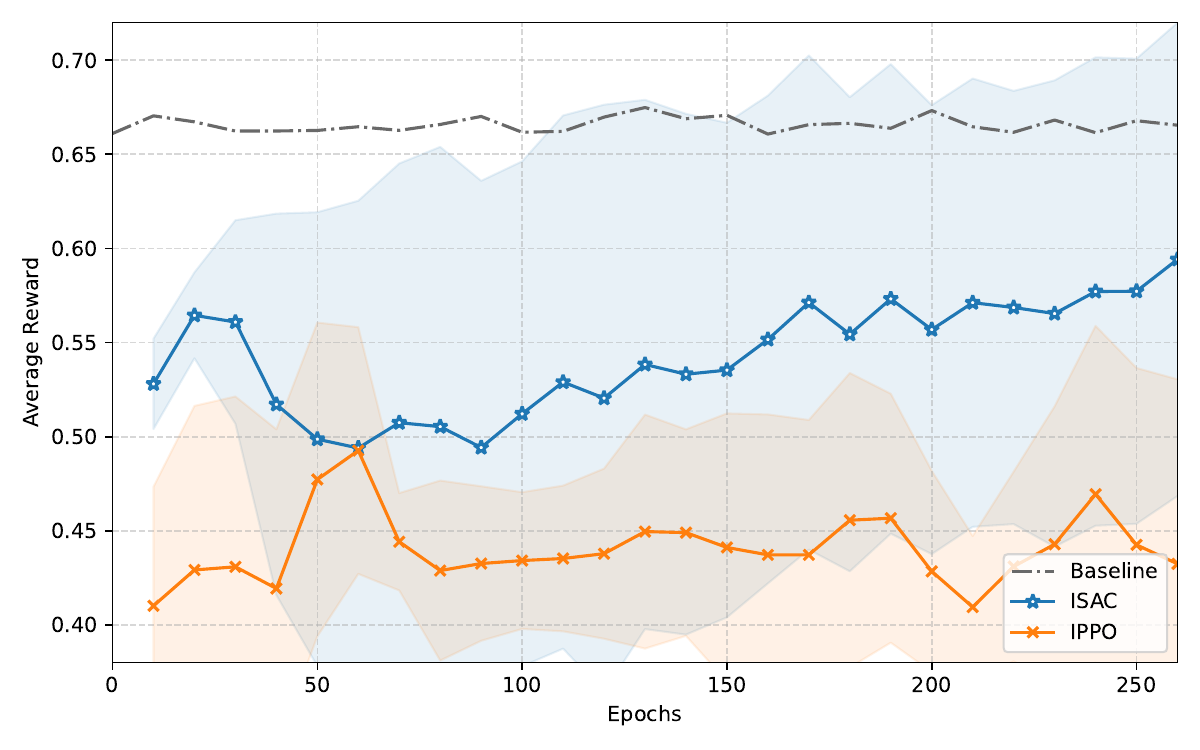}}
            \vspace{-0.5em}
		\caption{{Performance of IRL without communication.}} 
            \vspace{-0.6em}
		\label{fig:IRL}
	\end{center}
\end{figure*} 

\begin{figure*}[tbp]
	\subfigcapskip = -1pt
	\begin{center}
		\subfigure[{``Figure 8"}]
		{\includegraphics[scale =0.38]{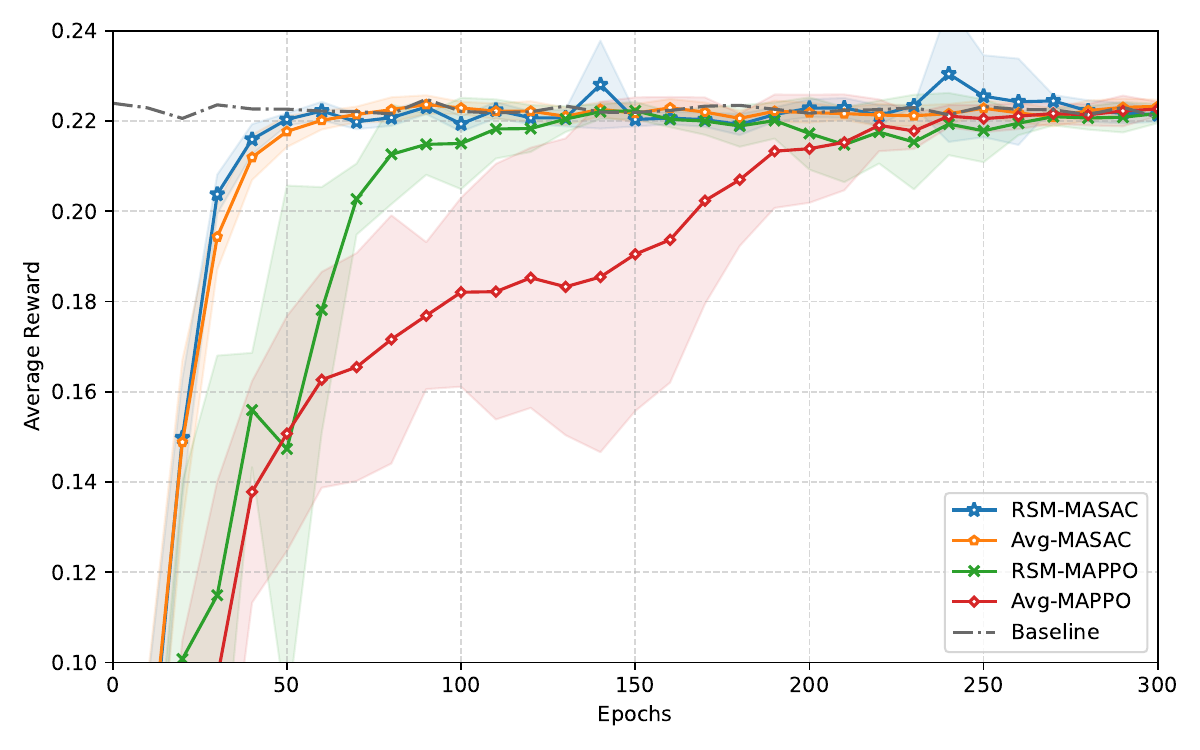}}
  \hspace{0.8cm}
		\subfigure[{``Merge"}]
		{\includegraphics[scale =0.379]{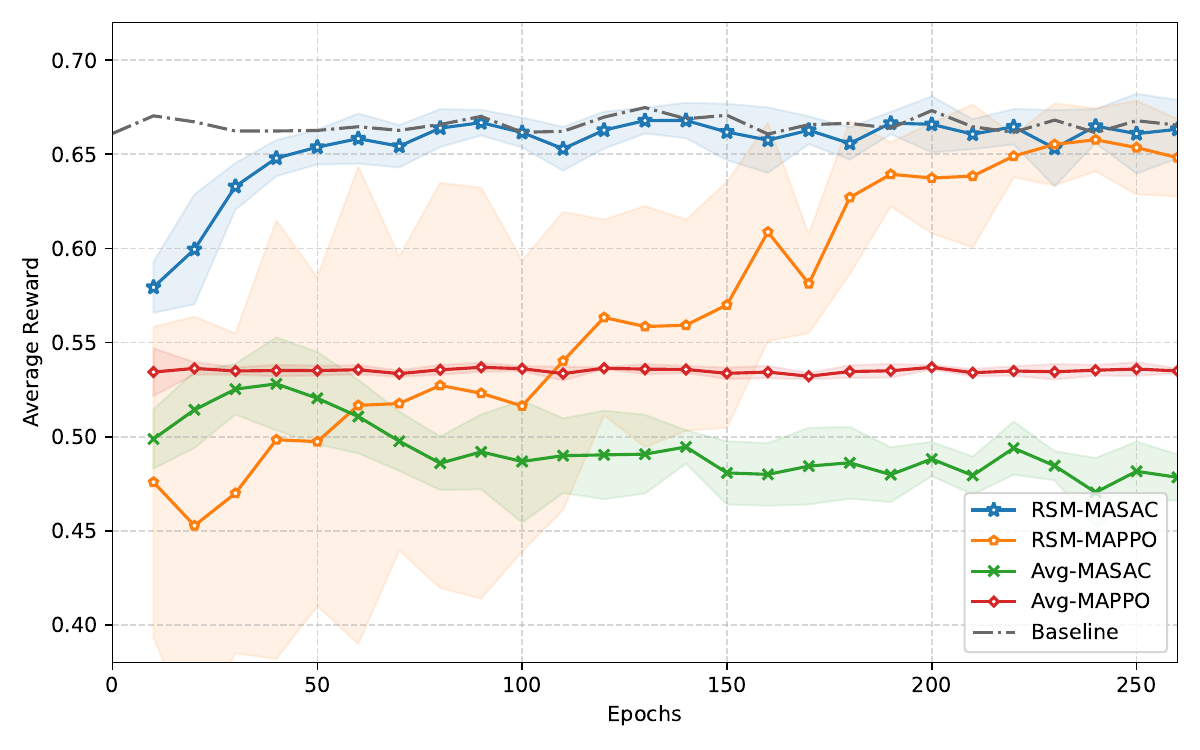}}
            \vspace{-0.5em}
		\caption{{Performance comparison of different methods.} }
  \label{fig:different method}
	\end{center}
 \vspace{-1.5em}
\end{figure*} 

For each referential policy, parameter mixing occurs only when $\mathbb{A}^{+}_{\pi}(\tilde{\pi})>0$. 
In this case, the computational complexity is predominantly governed by the calculation of $F(\theta)$ due to the outer product of policy gradients. 
Using an AMD EPYC 7F32 $@3.70$ GHz, $8$-core CPU and an NVIDIA RTX3090 GPU\footnote{We note that there exist available automotive chips like the NVIDIA DRIVE Orin or Thor featuring CUDA Tensor Core GPU.}, as shown in Fig. \ref{fig:time}, the computational time escalates with increasing sample size and parameter size, consistent with the discussions in Section \ref{sec:complexity}. 
However, with the aid of GPU hardware acceleration, such as superior parallel processing capabilities, higher memory bandwidth, and optimized deep learning libraries, the runtime experiences a significant reduction.

In order to reduce the occurrence of collisions and promote the traffic flow to the maximum desired speed, we take the normalized average speed of all vehicles at each timestep as the individual reward in each scenario\footnote{Notably, we assume complete knowledge of individual vehicle speeds at each vehicle here. Beyond the scope of this paper, some value-decomposition methods like \cite{xiao_stochastic_2023} can be further leveraged to derive a decomposed reward, so as to loosen such a strict requirement.}, which is assigned to each training agent after its action is performed. 
In addition, the current epoch will be terminated once a collision occurs or the max length of step $T$ in an epoch is reached. 

The results with variance are averaged over $5$ independent simulations, with tests conducted every $10$ epochs during training. 
Moreover, as for the baseline, we take the vehicles controlled by the Flow IDM \cite{treiber2000congested}, which belongs to a typical car-following model incorporating extensive prior knowledge and indicates the pinnacle of performance achievable by the best centralized federated MARL algorithms \cite{xu2021gradient}. 
Furthermore, the principal hyper-parameters used in simulations are listed in Table \ref{tab:parameters}.

\vspace{-0.1cm}
\subsection{Evaluation Metrics}\label{5Eval}
Apart from the average reward, we adopt some additional metrics to extensively evaluate the communication efficiency of RSM-MASAC. 
\begin{itemize}
    \item We denote the total count of reconstructed referential policies (i.e., all model replicas) as $\rho_\text{total}$. The number of effectively reconstructed referential policies, those contributing to the mixing process, is represented by $\rho_\text{ef}$. The mixing rate $\rho_r = \rho_\text{ef}/\rho_\text{total}$ thus reflects the usage rate of reconstructed policies. 
    \item We use $\psi$ to indicate the overall communication overhead (in terms of $\upsilon$) in an epoch, and $C_0$ denotes the number of communication rounds. 
    Therefore, the communication overheads equal $\psi = C_0\times c(\upsilon,f) = \frac{\rho_\text{total} \times \upsilon}{8\times 1024^3}$ (GB). 
\end{itemize}

\begin{figure*}[tbp]
    \centering
        \begin{minipage}[b]{0.39\textwidth}
        \centering
        \includegraphics[width=\textwidth]{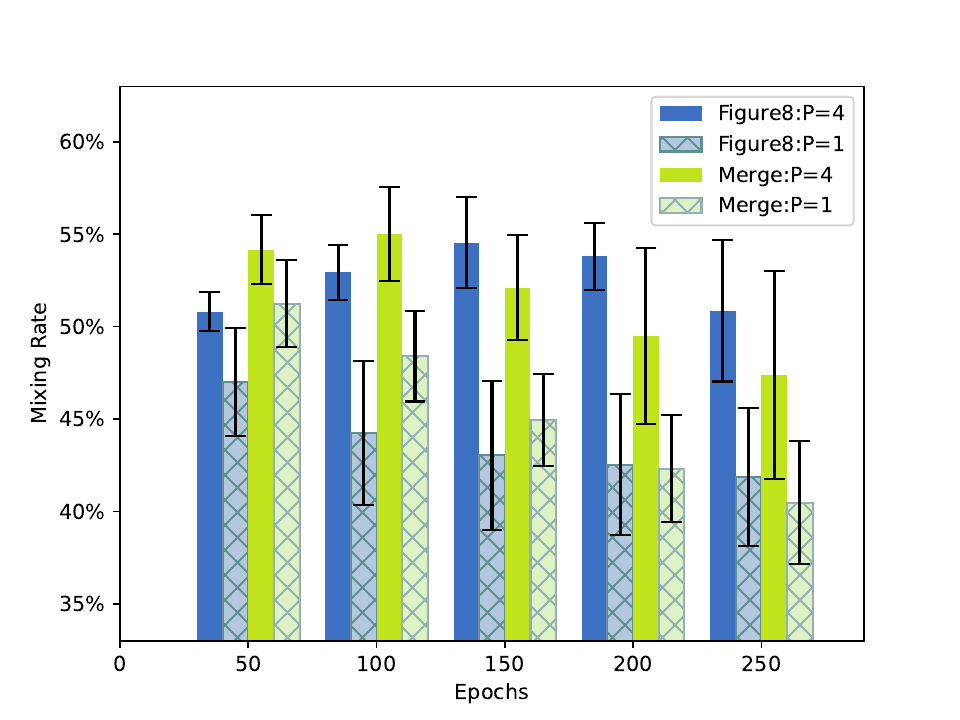}
        \caption{The effect of segmentation on mixing rate.}
        \label{fig:mixrate_column}
    \end{minipage}
    \begin{minipage}[b]{0.59\textwidth}
        \centering
        \begin{center}
		\subfigure[Communication times]
		{\includegraphics[scale =0.29]{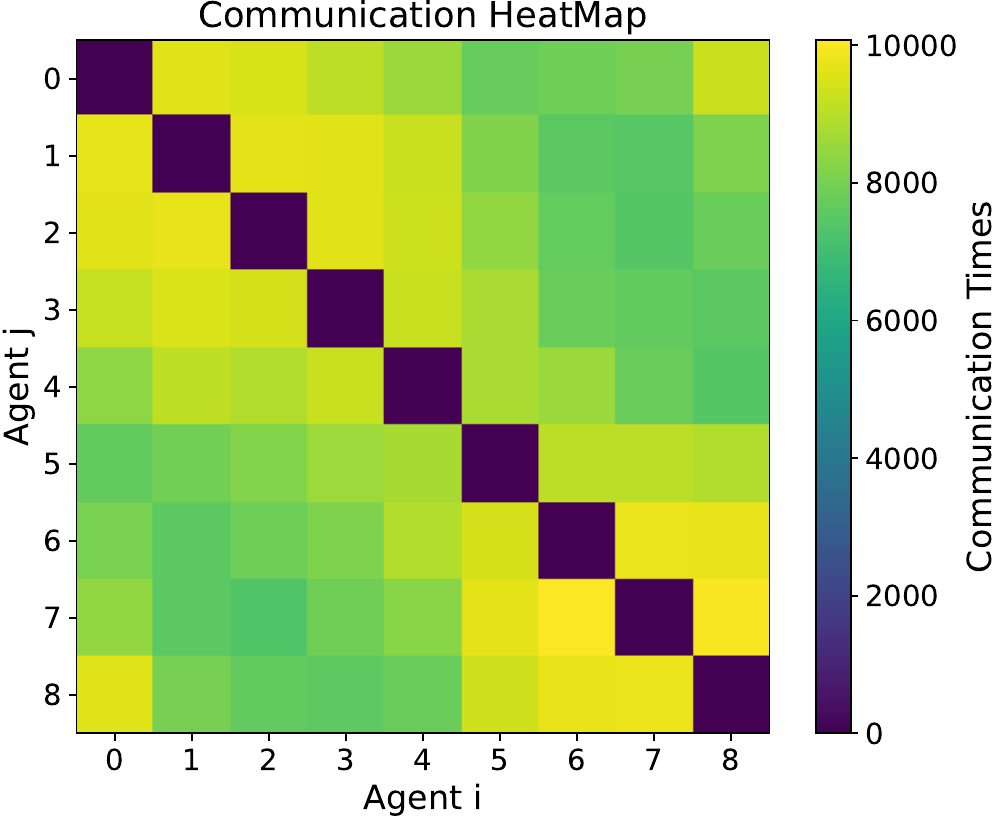}} 
        \hspace{0.02\textwidth}
		\subfigure[Effective mixture times]
		{\includegraphics[scale =0.29]{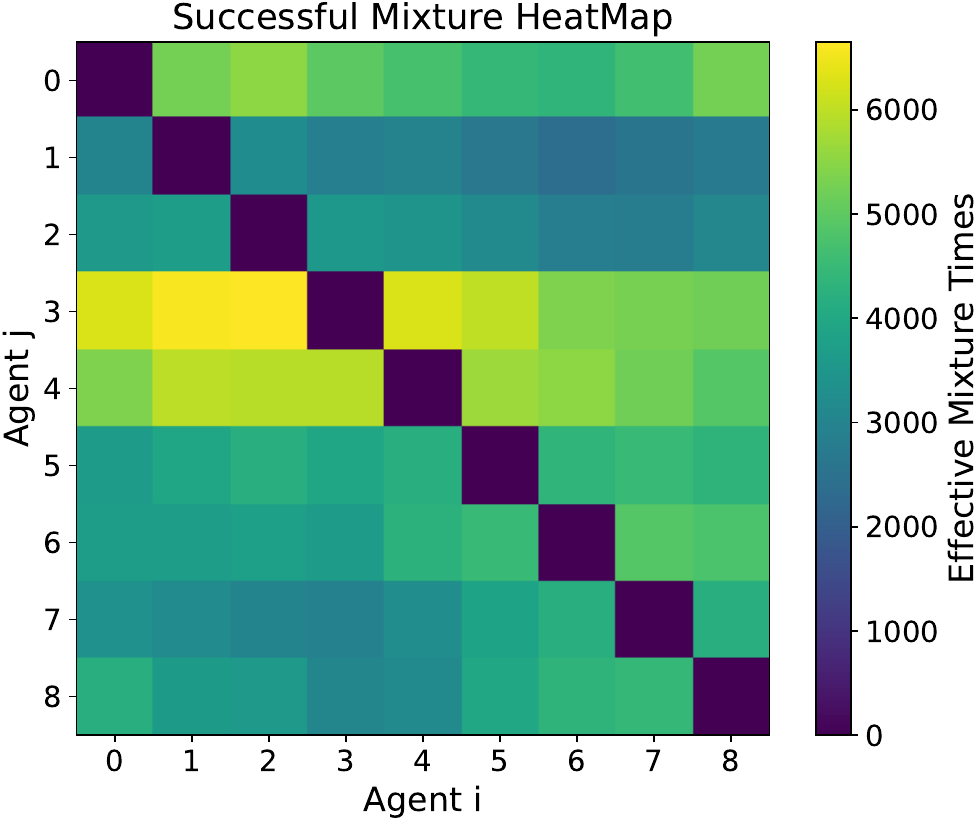}}
        \hspace{0.1cm}
		\caption{Heatmap comparison of communication times and effective mixture times under ``Figure 8''.}
        \label{fig:Heatmap}
	\end{center}
    \end{minipage}
\end{figure*}

\vspace{-0.1cm}
\subsection{Simulation Results}\label{5C}

\begin{table*}[htbp]
    \centering
    \caption{{Results of Average Reward \& Communication Efficiency under ``Merge".}}
    \label{tab:results}
    
    \begin{tabular}{l|Sc|Sc|Sc|Sc|Sc|Sc|Sc|Sc|Sc}
    \toprule
        \multicolumn{2}{c|}{Method} & $U$ & $P$ & $\kappa$ & Average Reward & $\rho_\text{total}$ & $\psi$ (GB) & $\rho_\text{ef}$ & $\rho_r$ \\ 
        \midrule
        \multirow{2}{*}{Avg-} & MAPPO & \multirow{10}{*}{$8$}  & \multirow{7}{*}{$4$} &  \multirow{3}{*}{$3$} & {$0.5351\pm0.0031$} & $1,626$ & $0.412$ & $1,626$ & $100\%$ \\ 
        \cline{2-2} \cline{6-10} 
        ~ & MASAC &  & ~ & ~ & {$0.4818\pm0.0161$} & $29,270$ & $7.425$ & $29,270$ & $100\%$ \\ 
        \cline{1-1} \cline{2-2}  \cline{6-10} 
        \multirow{10}{*}{RSM-} & MAPPO &  & ~ & ~ & $0.6528\pm0.0199$ & $1,546$ & $0.392$ & $711$ & $45.99\%$ \\ 
        \cline{2-2}  \cline{5-10} 
        ~ & \multirow{9}{*}{MASAC} & & ~ & $1$ & $0.6538\pm0.0250$ & $10,223$ & $2.593$ & $5,213$ & $50.99\%$ \\ 
         \cline{5-10} 
        ~ & ~ &  & ~ & $5$ & $0.6636\pm0.0115$ & $41,333$ & $10.486$ & $20,410$ & $49.38\%$ \\ 
         \cline{5-10} 
        ~ & ~ &  & ~ & $7$ & $0.6574\pm0.0156$ & $38,120$ & $9.670$ & $19,722$ & $51.74\%$ \\ 
         \cline{5-10} 
        ~ & ~ &  & ~ &  \multirow{6}{*}{$3$} & $0.6614\pm0.0165$ & $30,486$ & $7.734$ & $13,827$ & $45.36\%$ \\ 
        \cline{4-4}  \cline{6-10} 
        ~ & ~ &  & $1$ & ~ & $0.6535\pm0.0253$ & $30,066$ & $7.627$ & $12,117$ & $40.30\%$ \\ 
        \cline{4-4} \cline{6-10} 
        ~ & ~ &  & $2$ & ~ & $0.6206\pm0.0793$ & $29,634$ & $7.518$ & $14,549$ & $49.10\%$ \\ 
        \cline{4-4} \cline{6-10} 
        ~ & ~ &  & $6$ & ~ & $0.6505\pm0.0282$ & $30,126$ & $7.643$ & $13,902$ & $46.15\%$ \\ 
        \cline{3-3}\cline{4-4} \cline{6-10} 
        ~ & ~ & $72$ &  \multirow{2}{*}{$4$} & ~ & $0.6712\pm0.0282$ & $2,351$ & $0.596$ & $1,357$ & $57.72\%$ \\ 
        \cline{3-3} \cline{6-10} 
        ~ & ~ & $144$ & ~ & ~ & $0.6576\pm0.0269$ & $1,227$ & $ 0.311$ & $708$ & $57.70\%$ \\ 
        \bottomrule
    \end{tabular}
\end{table*}

Beforehand, we present the performance of Independent SAC (ISAC) and Independent PPO (IPPO) for MARL without any information sharing in Fig. \ref{fig:IRL}, so as to highlight the critical role of inter-agent communication in decentralized cooperative MARL and facilitate subsequent discussions. Fig. \ref{fig:IRL} reveals that the learning process of these non-cooperative algorithms achieves unstable average reward with greater variance, and suffers from convergence issues. 
Whereas, when a communication-assisted mixing phase involving the exchange of policy parameters is incorporated, as shown in Fig. \ref{fig:different method}, the learning process within the multi-agent environment is remarkably enhanced in terms of stability and efficiency. 
In other words, consistent with our previous argument, integrating DFL into IRL significantly improves training efficiency and ensures learning stability. 

In Fig. \ref{fig:different method}, we also compare the performance of Regulated Segment Mixture (RSM) and that of the direct, unselective Averaging (Avg) method, as mentioned in Section \ref{sec:system model and problem formulation}. 
Note that to speed up the calculation, we use a batch average gradient to approximate the calculation of $F(\theta)$ here. 
In Fig. \ref{fig:different method}, it can be observed that in terms of convergence speed and stability, the simple average mixture method is somewhat inferior. 
Particularly in more complex scenarios, as in ``Merge", which involves a greater number of RL agents and denser traffic flow compared to the ``Figure 8", the disparity between these two methods becomes more pronounced. The larger variance in the average reward under the average mixture method suggests a more unstable parameter mixing process, while the improvement in reward value implies the efficiency of communication information sharing. 
Therefore, it validates the effectiveness of RSM and supports the derived theoretical results for selecting useful reference policies and assigning appropriate mixing weights. 

\begin{figure*}[htbp]
	\subfigcapskip = -1pt
	\begin{center}
		\subfigure[predefined number of segments $P$]
		{\includegraphics[scale =0.45]{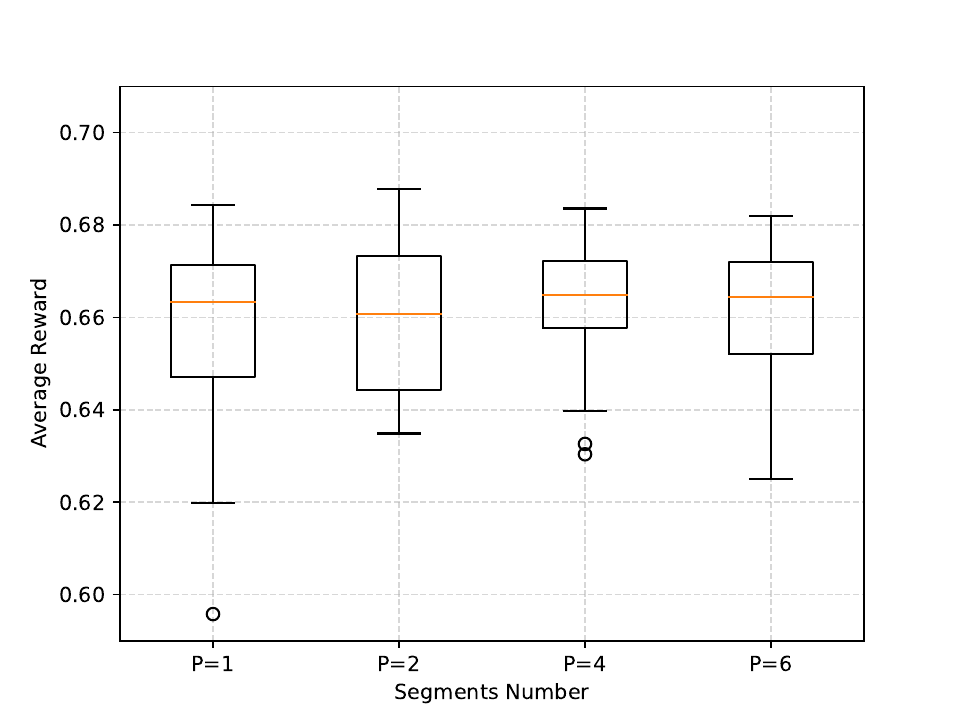}} 
		\subfigure[predefined number of replicas $\kappa$]
		{\includegraphics[scale =0.45]{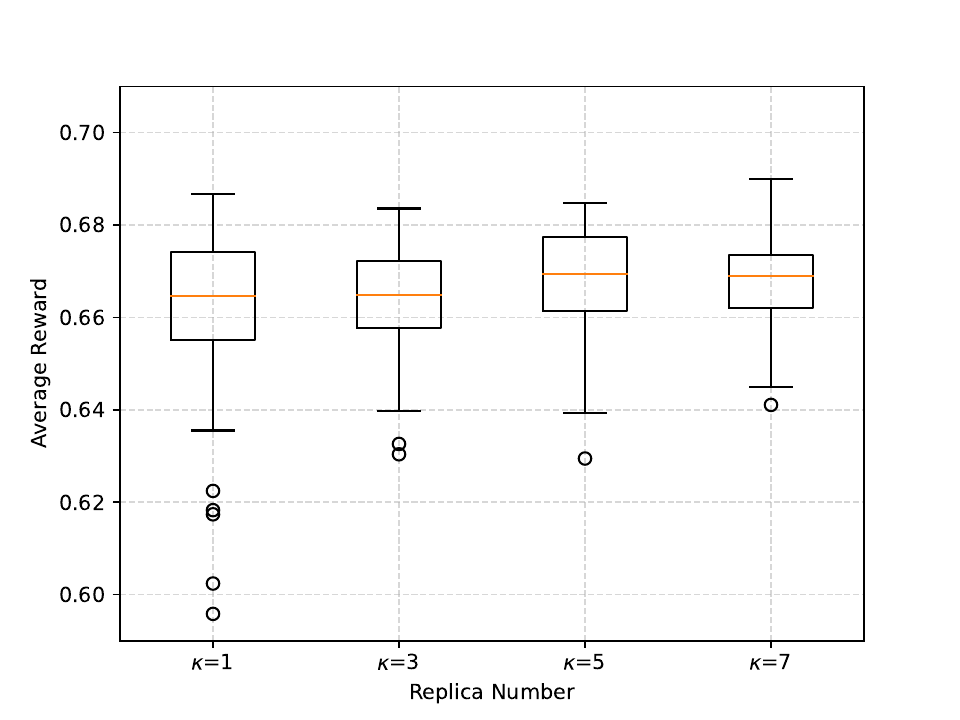}}
		\caption{The performance of RSM-MASAC with respect to different hyperparameters ((a) predefined number of segments $P$ and (b) predefined number of replicas $\kappa$), based on the testing episodes after $130$, $140$, $\cdots$, $200$ training epochs under ``Merge".} 
        \label{fig:Boxplot}
	\end{center}
\end{figure*} 

\begin{figure*}[tbp]
    \centering
    \begin{minipage}[b]{0.425\textwidth}
        \centering
        \includegraphics[width=\textwidth]{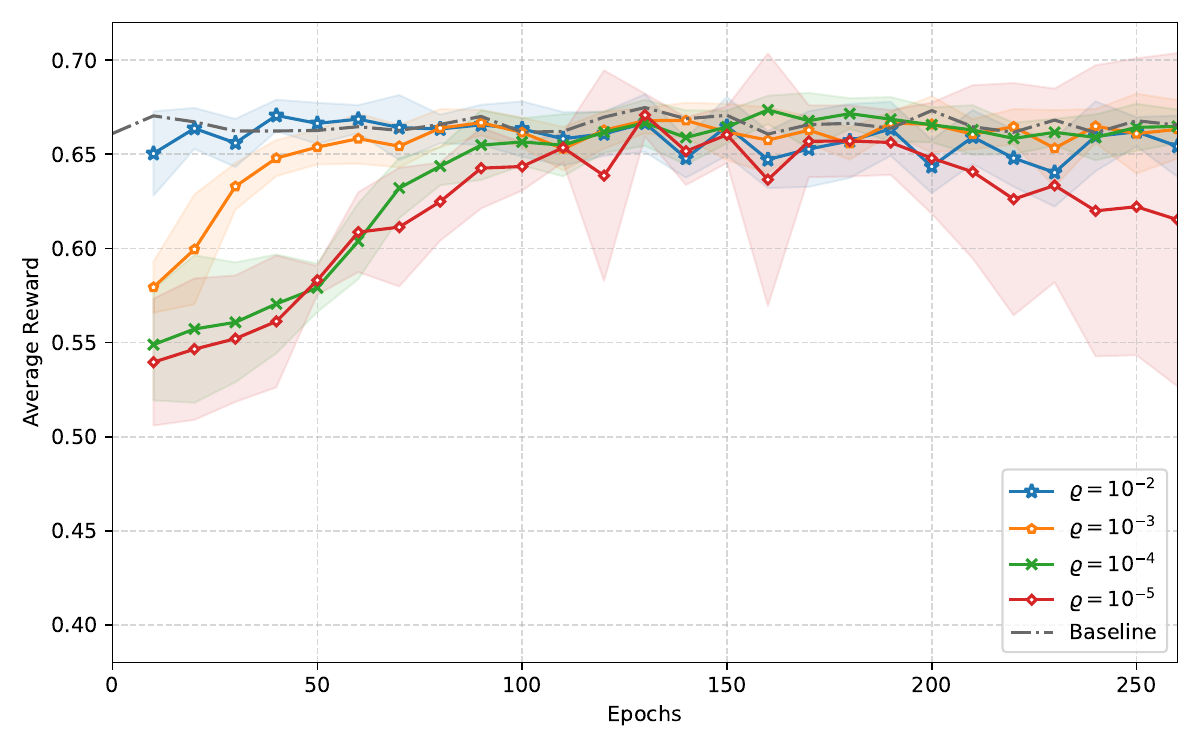}
        \caption{{Impact of different target smoothing coefficient $\varrho$ under ``Merge".}}
        \label{fig:target smoothing coefficient}
    \end{minipage}
    \hspace{0.8in}
     \begin{minipage}[b]{0.42\textwidth}
        \centering
        \includegraphics[width=\textwidth]{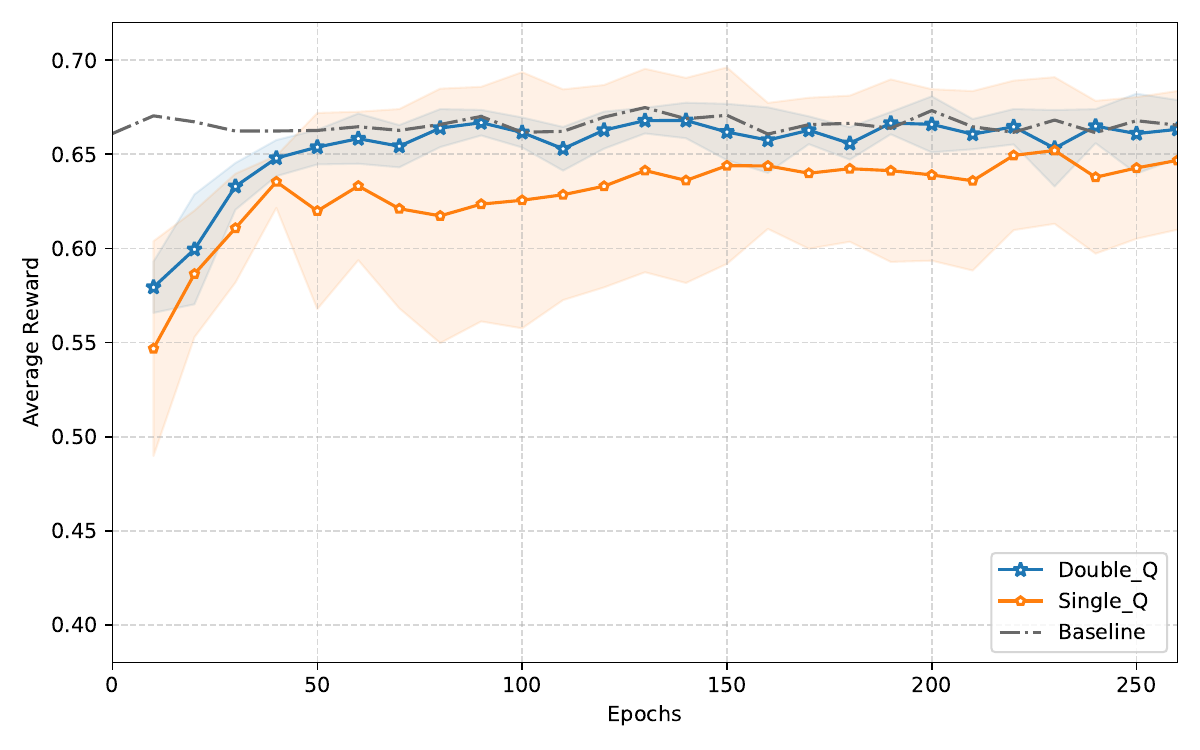}
		\caption{{Double vs single $Q$ network under ``Merge".} }
		\label{fig:whether double Q}
    \end{minipage}
\end{figure*}
\begin{figure}[tbp]
	\subfigcapskip = -1pt
	\begin{center}
  \hspace{-0.8cm} 
		{\includegraphics[scale =0.378]{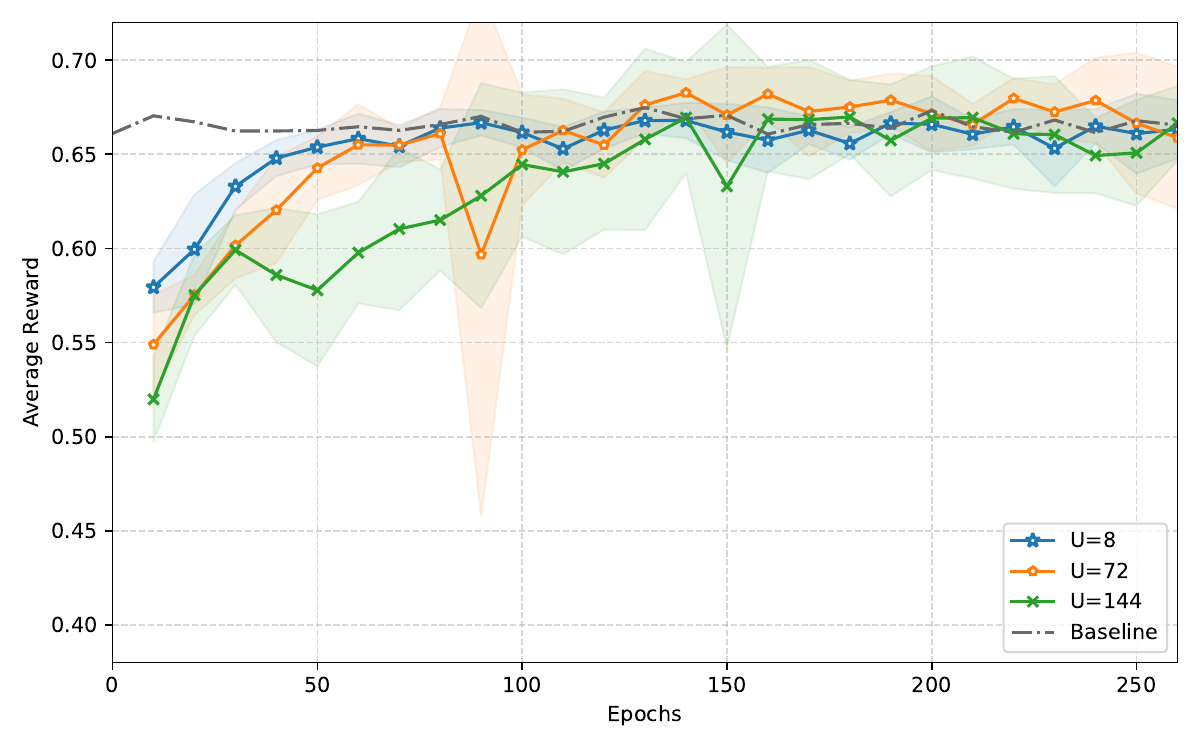}}
		\caption{{Impact of different communication intervals $U$ under ``Merge".}}
        \label{fig:different communication intervals}
	\end{center}
  \vspace{-1em}
\end{figure}

On the other hand, Fig. \ref{fig:different method} offers a further evaluation of the RSM's performance on both MASAC and MAPPO under the same policy DNN structure, learning rate and gradient clipping condition. 
In particular, as proposed in our previous work \cite{yu_communicationefficient_2023}, RSM-MAPPO utilizes PPO, which features lower sampling efficiency and policy update frequency, as the independent local learning algorithm under the traditional MARL framework, making it a special case of RSM-MASAC with $\alpha=0$ and a different policy advantage estimation under our more general reanalyses. 
Evidently, in both scenarios, RSM-MASAC demonstrates faster convergence and overall superior performance compared to RSM-MAPPO. 
We believe this is primarily attributed to the differences in exploration mechanism and sample efficiency. 
Specifically, benefiting from the introduced entropy item under MERL framework, SAC encourages policies to better explore environments, thus capably avoiding local optima and potentially learning faster. 

Next, we evaluate the impact of regulated segmentation. 
Fig. \ref{fig:mixrate_column} shows that the incorporation of segmentation (i.e., $P=4$) also leads to an improvement in the successful mixing rate $\rho_\text{r}$ of the policy, apart from better utilizing the available bandwidth as discussed in Section \ref{4B2}. 
The improvement lies in the introduction of a certain level of randomness, consistent with the idea to encourage exploration by adding an entropy term in MERL. But this randomness is also regulated by the policy parameter mixing theorem in Theorem \ref{theorem: parameter mix} without compromising performance. 
To embody the selection process of reconstructed referential policies, we further present the heatmap comparison between communication times and successful mixture times in Fig. \ref{fig:Heatmap}. 
It can be observed from Fig. \ref{fig:Heatmap}(a) that agents communicate frequently with their neighboring agents, but 
since not all communication packets (i.e., model segments) are utilized for improving local policy performance, it leads to the uneven and asymmetrical heatmap for the successful mixture of segments in Fig. \ref{fig:Heatmap}(b).
Our derived performance improvement bound and practical parameter mixture metric regulate these segments, demonstrating that the data from different agents contribute variably to the mixture.

In Fig. \ref{fig:Boxplot}, we further examine the impact of varying the predefined segmentation granularity $P$ and replicas $\kappa$ under RSM-MASAC. Notably, for each agent, the actual segment number and replicas at each time are also affected by the number of neighbors in the communication range, as addressed in Section \ref{4B1}.
The boxplots display the distribution of average reward collected during testing epochs ranging from $130$ to $200$, a period marked by increased fluctuations and slower convergence. 
These outcomes are inherently subject to the stochastic nature of the environment interacting with the RL agents and here we only indicate a rough trend that the median average reward shows a modest increase with a greater number of segments in Fig. \ref{fig:Boxplot}(a), and a slight upward trend with an increase replicas in Fig. \ref{fig:Boxplot}(b). 
However, these variations are minimal and contingent upon specific environmental conditions that are not consistent, thus warrant further discussions to balance communication bandwidth and computational overhead in practical applications. 
\begin{figure*}[tbp]
    \centering
    \begin{minipage}[b]{0.455\textwidth}
        \centering
        \includegraphics[width=\textwidth]{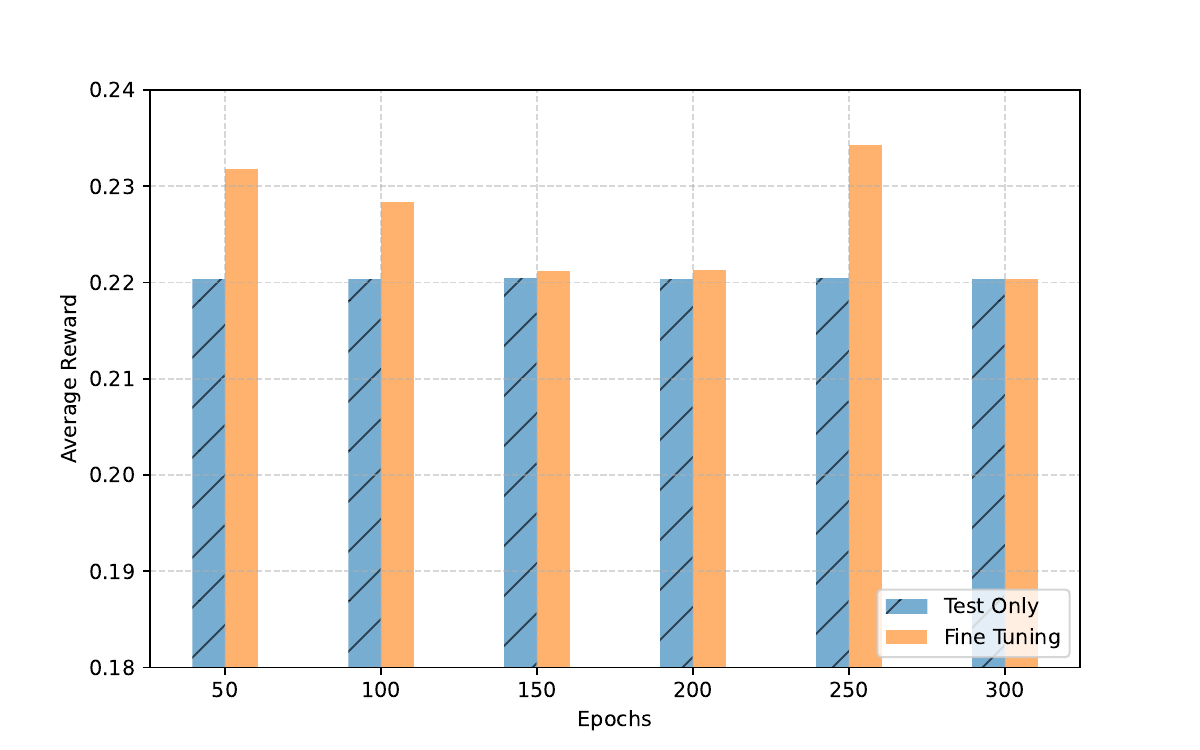}
		\caption{{Performance improvement of continually fine-tuning under ``Figure 8".}}  
		\label{fig:fine_tuning}
    \end{minipage}
    \hspace{0.4in}
     \begin{minipage}[b]{0.415\textwidth}
        \centering
        \includegraphics[width=\textwidth]{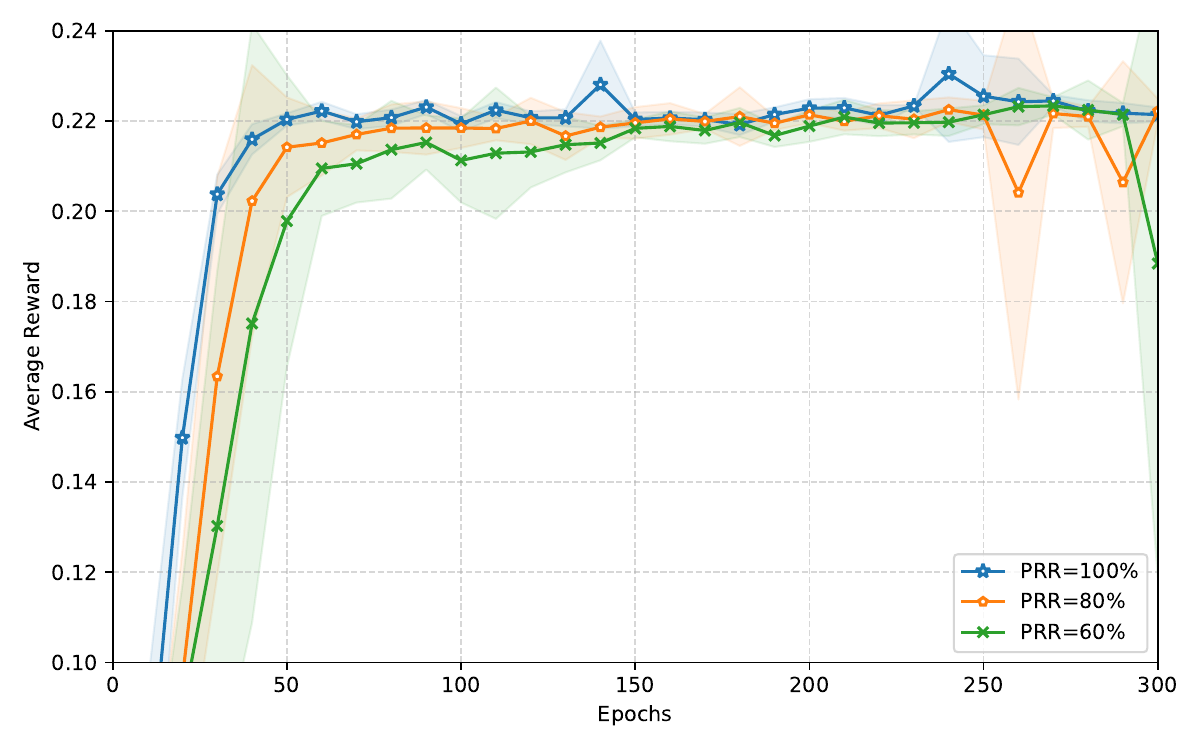}
		\caption{{Performance comparison of different PRR in communication under ``Figure 8".}}  
		\label{fig:PacketReceptionRate}
    \end{minipage}
\end{figure*}
In addition, we take the average reward within the last five testing epochs as the final converged performance, and summarize corresponding results and details about communication overhead in Table \ref{tab:results}, which takes the V2V communication examples' settings in \ref{4B2}. 

In addition, we conduct more ablation studies to evaluate the design of $Q$ network on the learning process. 
In Fig. \ref{fig:target smoothing coefficient}, we analyze the impact of different target smoothing coefficients $\varrho$ on the double $Q$-learning process. It can be observed that a larger value of the target smoothing coefficient $\varrho$, such as $10^{-2}$, can significantly accelerate the learning, but it also introduces disturbances, leading to non-steady learning. On the other hand, a smaller value like $10^{-4}$ or $10^{-5}$ noticeably slows down the learning speed and also introduces instability, affecting the performance of the algorithm. 
Moreover, in Fig. \ref{fig:whether double Q}, we also evaluate the adoption of dual $Q$ networks. It can be observed that employing dual $Q$ networks indeed accelerates the learning process. 
Furthermore, we also discuss the impact of communication intervals on the learning process, as shown in Fig. \ref{fig:different communication intervals}. It is evident that larger communication intervals result in reduced learning speeds and significant instability in the learning process, the same as classical FL processes. 

Meanwhile, we also investigate the effect of model fine-tuning in decentralized IoV speed control settings. 
Specifically, in a new environment full of $14$ DRL-controlled CAVs, we compare the performance between a continually fine-tuning policy and a well-trained, convergent policy (i.e., RSM-MASAC in Fig. \ref{fig:different method}) to execute inference only (i.e., test only in DRL description). 
The results in Fig. \ref{fig:fine_tuning} show that the fine-tuned policy in new conditions can reduce the occurrence of collisions at the beginning and continuously achieve better performance than just testing. Furthermore, to investigate the impact of possible packet loss when communicating on D2D or V2V links, we provide the result for different Packet Reception Rate (PRR) in Fig. \ref{fig:PacketReceptionRate}. 
It can be seen that a packet loss of less than $40\%$ primarily affects the convergence speed at the beginning, but produces trivial differences on the converged performance, demonstrating that DNN parameter transmission can yield good robustness in most real communication environments.

\section{Conclusions \& Future Work}\label{sec:conclusion}

In this paper, we have proposed a communication-efficient algorithm RSM-MASAC as a promising solution to enhance communication efficiency and policy collaboration in distributed MARL, particularly in the context of highly dynamic environments. 
By delving into the policy parameter mixture function, RSM-MASAC has provided a novel means to leverage and boost the effectiveness of distributed multi-agent collaboration. In particular, RSM-MASAC has successfully transformed the classical means of complete parameter exchange into segment-based request and response, which significantly facilitates the construction of multiple referential policies and simultaneously captures enhanced learning diversity. Moreover, in order to avoid performance-harmful parameter mixture, RSM-MASAC has leveraged a theory-established regulated mixture metric, and selects the contributive referential policies with positive relative policy advantage only. Finally, extensive simulations in the mixed-autonomy traffic control scenarios have demonstrated the effectiveness of the proposed approach. Notably, we use an approximate calculation here to reduce the computation complexity and speed up the calculation of $F(\theta)$, thus supporting the validation of the mixed performance improvement bound theorem. Despite its efficiency, its impact is still under investigation, and other methods to simplify FIM computations can be explored and substituted in the future.

\section*{Appendix: Proofs}
\subsection{Proof of Soft Policy Improvement}
\label{app:proof of soft policy improvement}
\begin{lemma}({Soft policy improvement})\label{lemma:appendix Soft policy improvement}
Let $\pi_\text{old}$ and $\pi_\text{new}$ be the optimizer of the minimization problem defined in \eqref{eq:SAC policy KL}. Then $Q^{\pi_\text{new}}(s,a)\geq Q^{\pi_\text{old}}(s,a), \forall s,a $ with $\vert \mathcal{A}\vert <\infty$. 
\end{lemma}

\begin{proof}
This proof is a direct application of soft policy improvement \cite{sac1, sac2}. We leave the proof here for completeness. 

Let $\pi_\text{old}\in \Pi$ and $Q^\pi,V^\pi$ is the corresponding soft state-action value and soft state value, respectively. And $\pi_\text{new}$ is defined as:
\begin{align*}
\pi_{\mathrm{new}}&=\arg\min_{\pi\in\prod} \mathrm{D}_{\mathrm{KL}}\left( \pi(\cdot\vert s)\Vert \frac{ \exp\left( \frac{1}{\alpha} Q^{\pi_\text{old}}({s},\cdot)\right)} {Z^{\pi_\text{old}}(s)}\right)\nonumber\\
&=\arg\min_{\pi\in\prod}J_{\pi_\text{old}}(\pi(\cdot\vert s))  {.}
\end{align*}
Since we can always choose $\pi_\text{new}=\pi_\text{old}\in\Pi$, there must be $J_{\pi_\text{old}}(\pi_\text{new}(\cdot\vert s))\leq J_{\pi_\text{old}}(\pi_\text{old}(\cdot\vert s))$. 
Hence
\begin{align*}
 &\mathbb{E}_{{a}\sim{\pi_\mathrm{new}}}\left[\frac{1}{\alpha}Q^{\pi_\text{old}}({s},{a}) 
 -
 \log \pi_\text{new}(a\vert s)
 -
 \log Z^{\pi_\text{old}}(s)\right]\\
 &\geq 
 \mathbb{E}_{{a}\sim\pi_\text{old}}\left[\frac{1}{\alpha}Q^{\pi_\text{old}}({s},{a})
 -
 \log \pi_\text{old}(a\vert s)
 -
 \log Z^{\pi_\text{old}}(s)\right]  {.}
\end{align*}
As partition function $Z$ depends only on the state, the inequality reduces to a form of the sum of entropy and value with one-step look-ahead
\begin{align*}
\mathbb{E}_{{a}\sim{\pi_\mathrm{new}}}&\left[Q^{\pi_\text{old}}({s},{a})\right]
    +
    \alpha H({\pi_\mathrm{new}}(\cdot\vert {s}))
     \\
    &\geq
    \mathbb{E}_{{a}\sim\pi_\text{old}}\left[Q^{\pi_\text{old}}({s},{a})\right] 
    +
    \alpha H(\pi_\text{old}(\cdot\vert {s}))
    = V_{\pi_\text{old}}(s)  {.}
\end{align*}
And according to the definition of the soft $Q$-value in section \ref{sec:Preliminary}, we can get that
\begin{align}
&Q^{\pi_\text{old}}(s,a) \nonumber\\
=& \mathop{\mathbb{E}}_{\substack{s_{1}}}\left[ r_0 
\!+\!
\gamma\!\left(\!\alpha H(\pi_\text{old}(\cdot\vert s_{1}))
\!+\!\!
\mathop{\mathbb{E}}_{\substack{a_{1}\sim \pi_\text{old}}} [Q^{\pi_\text{old}}(s_{1},a_{1})]\!\right)\!\right] \label{eq:A.1}\\
\leq & 
\mathop{\mathbb{E}}_{\substack{s_{1}}}\left[ r_0 
\!+\!
\gamma\!\left(\! \alpha H(
{\pi_\text{new}}(\cdot\vert s_{1}))
\!+\!
\mathop{\mathbb{E}}_{\substack{a_{1}\sim {\pi_\text{new}}}} [Q^{\pi_\text{old}}(s_{1},a_{1})]\!\right)\!\right] \nonumber\\
=&
\mathop{\mathbb{E}}_{\substack{s_{1}\\a_1\sim{\pi_\text{new}}}}\Big[ r_0 
+
\gamma\left(\alpha H({\pi_\text{new}}(\cdot\vert s_{1}))+r_{1}\right)\Big]
 \nonumber\\
 &\quad+ \gamma^2 \mathop{\mathbb{E}}_{\substack{s_{2}}}\left[\alpha H(\pi_\text{old}(\cdot\vert s_{2}))
 + 
 \mathop{\mathbb{E}}_{\substack{a_{2}\sim \pi_\text{old}}} 
 [Q^{\pi_\text{old}}(s_{2},a_{2})]\right]
 \nonumber\\
\leq & 
\mathop{\mathbb{E}}_{\substack{s_{1}\\a_1\sim{\pi_\text{new}}}}\left[ r_0 
+
\gamma\left(\alpha H({\pi_\text{new}}(\cdot\vert s_{1}))+r_{1}\right)\right]
 \nonumber\\
 &\quad +\gamma^2 \mathop{\mathbb{E}}_{\substack{s_{2}}}\left[\alpha H({\pi_\text{new}}(\cdot\vert s_{2}))
 + 
 \mathop{\mathbb{E}}_{\substack{a_{2}\sim {\pi_\text{new}}}} 
 [Q^{\pi_\text{old}}(s_{2},a_{2})]\right] \nonumber\\
 =& \!
\mathop{\mathbb{E}}_{\substack{s_{1}\\a_{1}\sim{\pi_\text{new}}\\s_{2}}}\!\!\left[ r_0 
\!+\!
\gamma\!\left(\alpha H({\pi_\text{new}}(\cdot\vert s_{1}))\!+\!r_{1}\right)
\!+\!
\gamma^2\!\left(\alpha H({\pi_\text{new}}(\cdot\vert s_{2}))\!+\!r_{2}\right)\right]
 \nonumber\\
 &\quad +
 \gamma^3 \mathop{\mathbb{E}}_{\substack{s_{3}}}\left[\alpha H({\pi_\text{new}}(\cdot\vert s_{3}))
 + 
 \mathop{\mathbb{E}}_{\substack{a_{3}\sim {\pi_\text{new}}}} 
 [Q^{\pi_\text{old}}(s_{3},a_{3})]\right]
 \nonumber\\
&\quad \cdots  \nonumber\\
\leq & \mathop{\mathbb{E}}\limits_{\tau_1\sim\pi_\text{new}}\left[
r_0+\sum_{t=1}^{\infty}\gamma^t \big(H({\pi_\text{new}}(\cdot\vert s_t))+r_t\big)
\right]  \nonumber \\
= & Q^{\pi_\text{new}}(s,a)  {.} \nonumber
\end{align}
\end{proof}

\subsection{Proof of Theorem \ref{theorem:mix policy improvement}}\label{app:theorem1}
Before proving the theorem, 
we first introduce several important lemmas, on which the proofs of the proposed theorems are built. 

\begin{lemma}\label{lemma:A_sum_zero}
\begin{equation}
    \mathbb{E}_{a\sim \pi}\left[A_\pi(s,a)\right]=-\alpha H(\pi(\cdot\vert s)) {.} \nonumber
\end{equation}
\end{lemma}

\begin{proof}
\begin{align}
    & \mathbb{E}_{a\sim \pi}\left[A_\pi(s,a)\right]\nonumber\\
    =&\sum_a \pi(a\vert s)A_\pi(s,a)\nonumber\\
    =&\sum_a \pi(a\vert s)\left[Q_\pi(s,a)-V_\pi(s)\right]\nonumber\\
    =&\sum_a \pi(a\vert s)Q_\pi(s,a)-V_\pi(s)\nonumber\\
    \stackrel{\mathrm{(a)}}{=}& \sum_a \pi(a\vert s)Q_\pi(s,a)-\sum_a \pi(a\vert s)\left[Q_\pi(s,a)-\alpha \log \pi(a\vert s)\right]\nonumber\\
    =&- \alpha H(\pi(\cdot\vert s))  {,} \nonumber
\end{align}
where the equality $(a)$ is according to \eqref{eq:V-function bellman equation}.
\end{proof}

\begin{lemma}\label{lemma:replace expectation}
\begin{equation}
    \mathop{\mathbb{E}}\limits_{a\sim \pi_\text{mix}}\left[A_\pi(s,a)\right]=\beta\mathop{\mathbb{E}}\limits_{a\sim\tilde{\pi}}\left[A_\pi(s,a)\right] -(1-\beta) \alpha H(\pi(\cdot\vert s)) {.}\nonumber
\end{equation}
\end{lemma}

\begin{proof}
\begin{align}
    & \mathbb{E}_{a\sim \pi_\text{mix}}\left[A_\pi(s,a)\right]\nonumber\\
    =&\sum_a \pi_\text{mix}(a\vert s)A_\pi(s,a)\nonumber\\
    \stackrel{\mathrm{(a)}}{=} &\sum_a \left[(1-\beta)\pi(a\vert s)+\beta \tilde{\pi}(a\vert s)\right]A_\pi(s,a)\nonumber\\
    = &\beta \sum_a \tilde{\pi}(a\vert s)A_\pi(s,a)+(1-\beta)\sum_a {\pi}(a\vert s)A_\pi(s,a)\nonumber\\
    \stackrel{\mathrm{(b)}}{=} &\beta\mathop{\mathbb{E}}\limits_{a\sim\tilde{\pi}}\left[A_\pi(s,a)\right]-(1-\beta) \alpha H(\pi(\cdot\vert s))  {,} \nonumber
\end{align}
where the equalities $(a)$ and $(b)$ are due to \eqref{eq:update rule} and  Lemma \ref{lemma:A_sum_zero}, respectively. 
\end{proof}

\begin{lemma}\label{lemma4:mix minus pi}
\begin{align*}
    &\eta(\pi_\text{mix})-\eta(\pi)\\
    =&
    \sum_{t=0}^\infty\gamma^t\left[
    \mathbb{E}_{\substack{s\sim P(s_t;\pi_\text{mix})\\a\sim \pi_\text{mix}}}  [A_\pi(s,a)+\alpha H(\pi_\text{mix}(\cdot\vert s))]\right] {.}\nonumber
\end{align*}
\begin{proof}
First note that according to the definition of advantage function and soft bellman equation in \eqref{eq:A definition} and \eqref{eq:V-function bellman equation}, we can get $A_\pi(s_t,a_t)=\mathbb{E}_{s_{t+1}\sim P(\cdot\vert s_t,a_t)}[r_t + \gamma V_\pi(s_{t+1})-V_\pi(s_t)]$. 
Therefore, 
\begin{align}
&\mathop{\mathbb{E}}\limits_{\tau_0\sim\pi_\text{mix}}
    \left[\sum_{t=0}^\infty \gamma^t \left[A_\pi(s_t,a_t)\right]\right]\nonumber\\
=&\mathop{\mathbb{E}}\limits_{\tau_0\sim\pi_\text{mix}}
    \left[\sum_{t=0}^\infty \gamma^t 
   \left[r_t + \gamma V_\pi(s_{t+1})-V_\pi(s_t) \right]\right]\nonumber\\
=&\mathop{\mathbb{E}}\limits_{\tau_0\sim\pi_\text{mix}}\Bigg[\sum_{t=0}^\infty \gamma^t r_t+ (\gamma V_\pi(s_1)-V_\pi(s_0)
    \nonumber\\
    & \quad\quad +\gamma^2 V_\pi(s_2)-\gamma V_\pi(s_1)+\cdots) \Bigg]\nonumber\\
=&\mathbb{E}_{\tau_0\sim\pi_\text{mix}}\left[\sum_{t=0}^\infty \gamma^t r_t\right]-\mathbb{E}_{s_0}\left[V_\pi(s_0)\right]
   \nonumber\\
   {=} & \mathbb{E}_{\tau_0\sim\pi_\text{mix}}\left[\sum_{t=0}^\infty \gamma^t (r_t+\alpha H(\pi_\text{mix}(\cdot\vert s_t))\right.\nonumber\\
   &\quad \left. -\sum_{t=0}\gamma^t[\alpha H(\pi_\text{mix}(\cdot\vert s_t))]\right]-\eta(\pi)
   \nonumber\\
   =& \eta(\pi_\text{mix})-\eta(\pi) - \mathbb{E}_{\tau_0\sim\pi_\text{mix}}\left[
   \sum_{t=0}^\infty\gamma^t[\alpha H(\pi_\text{mix}(\cdot\vert s_t))]\right]  {.} 
   \nonumber
\end{align}
Also we have
\begin{align*}
    &\mathop{\mathbb{E}}\limits_{\tau_0\sim\pi_\text{mix}}
    \left[\sum_{t=0}^\infty \gamma^t \left[A_\pi(s_t,a_t)+\alpha H(\pi_\text{mix}(\cdot\vert s_t))\right]\right]\\
    = &
    \sum_{t=0}^\infty \gamma^t \left[
    \mathbb{E}_{\substack{s\sim P(s_t;\pi_\text{mix})\\a\sim \pi_\text{mix}}}  [A_\pi(s,a)+\alpha H(\pi_\text{mix}(\cdot\vert s))]\right]  {.}
\end{align*}
Hence we can derive the formula for Lemma \ref{lemma4:mix minus pi}.
\end{proof}
\end{lemma}

\begin{lemma}\label{lemma:mix entropy}
For any given state $s$, we can decompose the entropy of the mixed policy as 
\begin{align*}
    H(\pi_{\text{mix}}(\cdot\vert s)) &= \mathrm{D}_{\mathrm{JS}}^\beta(\tilde{\pi}(\cdot\vert s)\Vert \pi(\cdot\vert s))+\beta H(\tilde{\pi}(\cdot\vert s))  \\
    &\qquad+(1-\beta) H(\pi(\cdot\vert s))  {.}
\end{align*}
\end{lemma}

\begin{proof} By definition, 
\begin{align*}
    &H(\pi_\text{mix}(\cdot\vert s)) \\
    =& -\sum_a \Big\{ [(1-\beta)\pi(a\vert s)+\beta\tilde{\pi}(a\vert s)]\\
    & \qquad \qquad \cdot \log[(1-\beta)\pi(a\vert s)+\beta\tilde{\pi}(a\vert s)] \Big\} \\
    =&\sum_a \beta \left[\tilde{\pi}(a\vert s)\log\frac{\tilde{\pi}(a\vert s)}{(1-\beta)\pi(a\vert s)+\beta\tilde{\pi}(a\vert s)}\right] 
    \\
    & \quad + \sum_a (1-\beta)\left[{\pi}(a\vert s)\log\frac{{\pi}(a\vert s)}{(1-\beta)\pi(a\vert s)+\beta\tilde{\pi}(a\vert s)}\right] 
    \\
    &\quad - \sum_a \left[\beta\tilde{\pi}(a\vert s)\log \tilde{\pi}(a\vert s) \right]
    \\
    & \quad - \sum_a \left[(1-\beta){\pi}(a\vert s)\log{\pi}(a\vert s) \right]
    \\
    =& \mathrm{D}_{\mathrm{JS}}^\beta(\tilde{\pi}(\cdot\vert s)\Vert \pi(\cdot\vert s)) +\beta H(\tilde{\pi}(\cdot\vert s)) +(1-\beta) H(\pi(\cdot\vert s))  {.}
\end{align*}
\end{proof}

\begin{lemma}\label{lemma3:mix A}
\begin{align*}
    &\mathop{\mathbb{E}}_{\substack{{s\sim P (s_t;\pi_\text{mix})}\\
    {a\sim \pi_\text{mix}}}}\left[A_\pi(s,a)+ \alpha H(\pi_\text{mix}(\cdot\vert s))\right]
    \\
    \geq &
    \quad \beta \mathop{\mathbb{E}}_{\substack{s\sim P (s_t;\pi)\\
    a\sim \tilde{\pi}}}\left[A_\pi(s,a)+\alpha H(\tilde{\pi}(\cdot\vert s))\right] 
    - 2\beta \rho_t \varepsilon\\
    &\quad +\alpha \mathop{\mathbb{E}}\limits_{s\sim P (s_t;\pi_\text{mix})}\left[D_{\text{JS}}^\beta(\tilde{\pi}(\cdot\vert s)\Vert \pi(\cdot\vert s))\right] {,}
\end{align*}
where $\varepsilon = \max_{s}\vert \mathbb{E}_{a\sim \tilde{\pi}}[A_\pi(s,a)+\alpha H(\tilde{\pi}(\cdot\vert s))]\vert$, and $\rho_t=1-(1-\beta)^t$.
\end{lemma}

\begin{proof}
\begin{align*}
     &\mathop{\mathbb{E}}_{\substack{{s\sim P (s_t;\pi_\text{mix})}\\
    {a\sim \pi_\text{mix}}}}\left[A_\pi(s,a)+ \alpha H(\pi_\text{mix}(\cdot\vert s))\right]\\
    &\stackrel{\mathrm{(a)}}{=}  \mathop{\mathbb{E}}\limits_{s\sim P (s_t;\pi_\text{mix})}
    \Bigg[\beta\sum_a \tilde{\pi}(a\vert s)A_\pi(s,a) -(1-\beta) \alpha H(\pi(\cdot\vert s)) \\
    & \quad + \alpha \mathrm{D}_{\mathrm{JS}}^\beta(\tilde{\pi}(\cdot\vert s)\Vert \pi(\cdot\vert s)) \!+\!\beta \alpha H(\tilde{\pi}(\cdot\vert s))
    \!+\!(1\!-\!\beta) \alpha H(\pi(\cdot\vert s))\Bigg]\\
    &=\beta \mathop{\mathbb{E}}\limits_{s\sim P (s_t;\pi_\text{mix})}\left[\sum_a \tilde{\pi}(a\vert s)\Big[A_\pi(s,a)+\alpha H(\tilde{\pi}(\cdot\vert s))\Big]\right]\\
    & \quad +\alpha \mathop{\mathbb{E}}\limits_{s\sim P (s_t;\pi_\text{mix})}\left[\mathrm{D}_{\mathrm{JS}}^\beta(\tilde{\pi}(\cdot\vert s)\Vert \pi(\cdot\vert s))\right]  {,}
\end{align*}
where the equality $(a)$ is according to Lemma \ref{lemma:replace expectation} and Lemma \ref{lemma:mix entropy}, and the mixed policy is taken as a mixture of the policy $\pi$ and the referential policy $\tilde{\pi}$ received from others. In other words, to sample from $\pi_\text{mix}$, we first draw a Bernoulli random variable, which tells us to choose $\pi$ with probability $(1-\beta)$ and choose $\tilde{\pi}$ with probability $\beta$. Let $c_t$ be the random variable that indicates the number of times $\tilde{\pi}$ was chosen before time $t$. $P(s_t;\pi)$ is the distribution over states at time $t$ while following $\pi$. We can condition on the value of $c_t$ to break the probability distribution into two pieces, with $P(c_t=0)=(1-\beta)^t$, and $\rho_t=P(c_t\geq 1)=1-(1-\beta)^t$. 
Thus, we can get \eqref{eq:mathopE_expansion} on Page \pageref{eq:mathopE_expansion}.
\begin{figure*}
\begin{align}
    &\beta \mathop{\mathbb{E}}\limits_{s\sim P (s_t;\pi_\text{mix})}\left[\sum_a \tilde{\pi}(a\vert s)\left[A_\pi(s,a)+\alpha H(\tilde{\pi}(\cdot\vert s))\right]\right]\nonumber\\
    =&\beta (1-\rho_t)\!\mathop{\mathbb{E}}\limits_{s\sim P (s_t\vert c_t= 0 ;\pi_\text{mix})}\!\left[\sum_a \tilde{\pi}(a\vert s)\left[A_\pi(s,a)+\alpha H(\tilde{\pi}(\cdot\vert s))\right]\right]+ \beta \rho_t\mathop{\mathbb{E}}\limits_{s\sim P (s_t\vert c_t\geq 1 ;\pi_\text{mix})}\left[\sum_a \tilde{\pi}(a\vert s)\left[A_\pi(s,a)+\alpha H(\tilde{\pi}(\cdot\vert s))\right]\right]\nonumber\\
    =& \beta\!\mathop{\mathbb{E}}\limits_{s\sim P (s_t\vert c_t= 0 ;\pi_\text{mix})}\!\left[\sum_a \tilde{\pi}(a\vert s)\left[A_\pi(s,a)+\alpha H(\tilde{\pi}(\cdot\vert s))\right]\right]- \beta \rho_t\mathop{\mathbb{E}}\limits_{s\sim P (s_t\vert c_t= 0 ;\pi_\text{mix})}\left[\sum_a \tilde{\pi}(a\vert s)\left[A_\pi(s,a)+\alpha H(\tilde{\pi}(\cdot\vert s))\right]\right]\label{eq:mathopE_expansion}\\
    &\quad+ \beta \rho_t\mathop{\mathbb{E}}\limits_{s\sim P (s_t\vert c_t\geq 1 ;\pi_\text{mix})}\left[\sum_a \tilde{\pi}(a\vert s)\left[A_\pi(s,a)+\alpha H(\tilde{\pi}(\cdot\vert s))\right]\right]\nonumber
\end{align}
\hrulefill
\end{figure*}
Recalling the definition $\varepsilon = \max_{s}\big\vert \mathbb{E}_{a\sim \tilde{\pi}}[A_\pi(s,a)+\alpha H(\tilde{\pi}(\cdot\vert s))]\big\vert$, we have 
\begin{align}
    &\mathop{\mathbb{E}}\limits_{s\sim P (s_t\vert c_t= 0 ;\pi_\text{mix})}\left[\sum_a \tilde{\pi}(a\vert s)\left[A_\pi(s,a)+\alpha H(\tilde{\pi}(\cdot\vert s))\right]\right]\leq\varepsilon\nonumber  {;} \\
    &\mathop{\mathbb{E}}\limits_{s\sim P (s_t\vert c_t\geq 1 ;\pi_\text{mix})}\left[\sum_a \tilde{\pi}(a\vert s)\left[A_\pi(s,a)+\alpha H(\tilde{\pi}(\cdot\vert s))\right]\right]\geq -\varepsilon\nonumber  {.}
\end{align}
Therefore, \eqref{eq:mathopE_expansion} can be re-organized as  
\begin{align}
    & \beta \mathop{\mathbb{E}}\limits_{s\sim P (s_t;\pi_\text{mix})}\left[\sum_a \tilde{\pi}(a\vert s)\left[A_\pi(s,a)+\alpha H(\tilde{\pi}(\cdot\vert s))\right]\right]\nonumber\\
    \geq &  \beta\!\mathop{\mathbb{E}}\limits_{s\sim P (s_t\vert c_t= 0 ;\pi_\text{mix})}\!\left[\sum_a \tilde{\pi}(a\vert s)\left[A_\pi(s,a)+\alpha H(\tilde{\pi}(\cdot\vert s))\right]\right]\nonumber\\
    &\quad- 2\beta \rho_t \varepsilon \nonumber\\
    = & \beta\!\mathop{\mathbb{E}}\limits_{s\sim P (s_t;\pi)}\!\left[\sum_a \tilde{\pi}(a\vert s)\left[A_\pi(s,a)+\alpha H(\tilde{\pi}(\cdot\vert s))\right]\right]\!-\! 2\beta \rho_t \varepsilon \nonumber  {,}
\end{align}
in which ${P (s_t\vert c_t= 0 ;\pi_\text{mix})}={P (s_t;\pi)}$. 
\end{proof}

Next, we are ready to prove Theorem \ref{theorem:mix policy improvement}.
\begin{proof}
According to Lemma \ref{lemma4:mix minus pi} and Lemma \ref{lemma3:mix A}, we have
\begin{align*}
    &\eta(\pi_\text{mix})-\eta(\pi) \\
    =&
    \sum_{t=0}^\infty \gamma^t \left[
    \mathbb{E}_{\substack{s\sim P(s_t;\pi_\text{mix})\\a\sim \pi_\text{mix}}}  [A_\pi(s,a)+\alpha H(\pi_\text{mix}(\cdot\vert s))]\right]\\
    \geq &
    \beta \sum_{t=0}^\infty \gamma^t \!\mathop{\mathbb{E}}_{\substack{s\sim P (s_t;\pi)\\
    a\sim \tilde{\pi}}}\left[A_\pi(s,a)+\alpha H(\tilde{\pi}(\cdot\vert s))\right] 
    - 2\beta \varepsilon \sum_{t=0}^\infty \gamma^t \rho_t \\
    &\quad +\alpha  \sum_{t=0}^\infty \gamma^t \mathop{\mathbb{E}}\limits_{s\sim P (s_t;\pi_\text{mix})}\left[\mathrm{D}_{\mathrm{JS}}^\beta(\tilde{\pi}(\cdot\vert s)\Vert \pi(\cdot\vert s))\right]\\
     =&
     \beta \mathop{\mathbb{E}}_{\substack{s\sim d_\pi\\
    a\sim \tilde{\pi}}}\left[A_\pi(s,a)+\alpha H(\tilde{\pi}(\cdot\vert s))\right] - 2\beta \varepsilon\sum_{t=0}\gamma^t [1-(1-\beta)^t]
    \\
    &\quad +\alpha  \mathop{\mathbb{E}}_{s\sim d_{\pi_\text{mix}}}\left[\mathrm{D}_{\mathrm{JS}}^\beta(\tilde{\pi}(\cdot\vert s)\Vert \pi(\cdot\vert s))\right]
     \\  
    =&
    \beta \mathop{\mathbb{E}}_{\substack{s\sim d_\pi\\
    a\sim \tilde{\pi}}}\left[A_\pi(s,a)+\alpha H(\tilde{\pi}(\cdot\vert s))\right]  
    - \frac{2\gamma\varepsilon \beta^2 }{(1\!-\!\gamma)(1\!-\!\gamma(1-\beta))} \\
    &\quad +\alpha  \mathop{\mathbb{E}}_{s\sim d_{\pi_\text{mix}}}\left[\mathrm{D}_{\mathrm{JS}}^\beta(\tilde{\pi}(\cdot\vert s)\Vert \pi(\cdot\vert s))\right]
     \\  
    \geq &
    \beta \mathop{\mathbb{E}}_{\substack{s\sim d_\pi\\
    a\sim \tilde{\pi}}}\left[A_\pi(s,a)+\alpha H(\tilde{\pi}(\cdot\vert s))\right] 
    - \frac{2\gamma\varepsilon \beta^2 }{(1-\gamma)^2} 
     \\
     &\quad +\alpha  \mathop{\mathbb{E}}_{s\sim d_{\pi_\text{mix}}}\left[\mathrm{D}_{\mathrm{JS}}^\beta(\tilde{\pi}(\cdot\vert s)\Vert \pi(\cdot\vert s))\right]  {.}
\end{align*}
We have the theorem.
\end{proof}

\subsection{Derivation of \eqref{eq:estimate policy advantage}}
\begin{lemma}
\label{lemma:estimate of policy advantage}
\begin{align*}
\mathbb{A}^{+}_\pi(\tilde{\pi}) 
&\approx 
\mathbb{E}_{s_t,a_t\sim\mathcal{D}}
\bigg[
\left(\frac{\tilde{\pi}_{\tilde{\theta}}(a_t\vert s_t)-\pi_\theta(a_t\vert s_t)}{\pi_t(a_t\vert s_t)}\right)
 \!\mathop{\min}_{\substack{x\in 1,2}}Q_{\omega_x}({s}_t,{a}_t)\nonumber\\
& \qquad + \alpha[H(\tilde{\pi}_{\tilde{\theta}}(\cdot\vert s_t)) - H({\pi}_{{\theta}}(\cdot\vert s_t)) ]\bigg] {.}
\end{align*}
\end{lemma}

\begin{proof}
Following the definition of policy advantage in \eqref{eq:definition of A}, 
\begin{align*}
&\mathbb{A}^{+}_\pi(\tilde{\pi}) \\
=&
\mathop{\mathbb{E}}_{\substack{s\sim d_\pi, a\sim \tilde{\pi}}}\left[A_\pi(s,a)+\alpha H(\tilde{\pi}(\cdot\vert s))\right]\\
\stackrel{\mathrm{(a)}}{=}&
\mathop{\mathbb{E}}_{\substack{s\sim d_\pi}}\left[\sum_a\tilde{\pi}(a\vert s)Q^\pi(s,a)-V^\pi(s)+\alpha H(\tilde{\pi}(\cdot\vert s))\right]\\
\stackrel{\mathrm{(b)}}{=}&
\mathop{\mathbb{E}}_{\substack{s\sim d_\pi}}\Big[\sum_a\tilde{\pi}(a\vert s)Q^\pi(s,a)- \sum_a{\pi}(a\vert s)Q^\pi(s,a)
\\
& \quad -\alpha H({\pi}(\cdot\vert s))+\alpha H(\tilde{\pi}(\cdot\vert s))\Big]  {,}
\end{align*}
where the equality $(a)$ is according to the state-action advantage value in \eqref{eq:A definition} while $(b)$ follows from \eqref{eq:V-function bellman equation}. 
Finally, by using Monte Carlo and importance sampling, we can get the lemma. 
    
\end{proof}

\bibliographystyle{IEEEtran}
\bibliography{reference}

\section*{Author Biographies}
\textbf{Xiaoxue Yu} is a PhD Candidate in Zhejiang University, Hangzhou, China. Her research interests currently focus on communications in distributed learning. 

\textbf{Rongpeng Li} is an Associate Professor in Zhejiang University. His research interests currently focus on networked intelligence. 

\textbf{Chengchao Liang} is a Full Professor with Chongqing University of Posts and Telecommunications, Chongqing, China. His research interests include wireless communications, satellite networks, Internet protocols, and optimization theory. 

\textbf{Zhifeng Zhao} is the Chief Engineer with Zhejiang Lab, Hangzhou, China. His research area includes collective intelligence and software-defined networks. 

\end{document}